\documentclass[12pt]{article}
\usepackage[utf8]{inputenc}
\usepackage[T1]{fontenc} 
\usepackage{fullpage}
\usepackage{setspace}

\usepackage[top=1.5in, bottom=1.5in, left=1.25in, right=1.25in, letterpaper]{geometry}  

\usepackage[indentafter]{titlesec}
\usepackage{microtype}
\usepackage{placeins}
\usepackage{enumerate}
\usepackage[shortlabels]{enumitem}
\usepackage{bm}
\usepackage[ruled,vlined,linesnumbered]{algorithm2e}
\usepackage{bbm}
\usepackage[colorinlistoftodos]{todonotes}

\usepackage{amsmath, amssymb, amsthm, amsfonts, latexsym, amscd, euscript, mathrsfs, mathtools}
\usepackage[colorlinks=true, linkcolor=darkred, citecolor=darkblue]{hyperref}
\usepackage{graphicx}
\definecolor{darkred}{RGB}{144,0,0}
\definecolor{darkblue}{RGB}{0,0,144}
\definecolor{lightyellow}{RGB}{255,255,224}
\usepackage[colorinlistoftodos]{todonotes}

\usepackage[retainorgcmds]{IEEEtrantools}
\usepackage{tabularx}
\usepackage{array}
\usepackage{caption}
\usepackage{xcolor}
\usepackage{subfig}
\usepackage{float}
\usepackage{tcolorbox}
\usepackage{tabularx}
\usepackage{threeparttable}
\usepackage{enumitem}
\usepackage{multirow}
\usepackage{lscape, pdflscape}
\usepackage[authoryear]{natbib}
\usepackage{fancyhdr}

\usepackage{booktabs}
\usepackage{array}
\usepackage{graphicx}
\usepackage{subcaption}
\newcolumntype{L}[1]{>{\raggedright\let\newline\\\arraybackslash\hspace{0pt}}m{#1}}
\newcolumntype{C}[1]{>{\centering\let\newline\\\arraybackslash\hspace{0pt}}m{#1}}
\newcolumntype{R}[1]{>{\raggedleft\let\newline\\\arraybackslash\hspace{0pt}}m{#1}}

\usepackage{accents}


\theoremstyle{definition} 
\theoremstyle{remark} \newtheorem{remark}{Remark}
\theoremstyle{plain} \newtheorem{theorem}{Theorem}
\theoremstyle{plain} 
\theoremstyle{plain} \newtheorem{lemma}{Lemma}
\theoremstyle{plain} 
\theoremstyle{plain} 
\theoremstyle{plain} \newtheorem{assumption}{Assumption}
\theoremstyle{plain} \newtheorem*{assumption*}{\assumptionletter}
\providecommand{\assumptionletter}{}
\makeatletter
\newcounter{bean}
\makeatother

\title{\bf{Robust Two-Sample Inference\\
under Serial Dependence}}
\author{Ulrich Hounyo\footnote{khounyo@albany.edu}
\\ University at Albany, SUNY
\\
\\ Min Seong Kim\footnote{min\_seong.kim@uconn.edu} 
\\ University of Connecticut }

\begin{document}
\newgeometry{margin=1in}
\maketitle



\begin{abstract}
\onehalfspacing
We propose robust inference for two-sample comparison with time-series data under serial dependence and heterogeneous long-run variances. Standardizing with orthonormal-basis (series HAR) projections, we develop two-sample $t$-tests and, for joint hypotheses on a vector of means, a series HAR Wald statistic. Because the increasing-$K$ chi-square limit tends to over-reject, we propose Welch-type fixed-$K$ $t$- and $F$-approximations with adjusted degrees of freedom. We further develop a series HAR wild bootstrap that reproduces serial dependence without resampling blocks. The framework nests difference-in-differences and Diebold--Mariano testing. Simulations and two empirical applications show accurate size control and competitive power, a tuning-free alternative to cluster-based inference.
\end{abstract}

\newpage

\onehalfspacing

\section{Introduction}
This paper proposes robust inference procedures for two-sample comparison with time-series data. Comparing two population means is a fundamental problem in data analysis. The classical two-sample $t$-test is perhaps the best-known example, typically introduced in introductory statistics courses. 
The classical two-sample $t$-test and its widely used extension due to \citet{Welch1947} both assume that the underlying observations are independent. This assumption is frequently violated in economic and financial applications, where data often exhibit serial dependence. Because serial dependence is ubiquitous in time series, two-sample $t$-tests that ignore it are generally invalid and can suffer from substantial size distortions.

We propose a general framework for conducting accurate and easy-to-use inference on the difference in population means---whether a scalar mean or a vector of means---when the data may display heterogeneity and complex, unknown forms of dependence. The framework builds on series-based heteroskedasticity and autocorrelation robust (HAR) long-run variance (LRV) estimators. Specifically, we form analogues of the classical two-sample $t$-test and Welch's test---together with their multivariate counterparts for joint hypotheses---by replacing the sample variances with the series HAR estimators. Implementation is straightforward. Each series is projected onto a set of orthonormal basis functions, and the sample variances of the resulting projection coefficients serve as LRV estimators.

Robust inference that accounts for heteroskedasticity and autocorrelation is standard practice in time-series econometrics and applied statistics. The most common approach estimates the LRV using kernel methods \citep{Newey_West1987, Andrews1991, Kiefer_Vogelsang2002,Kiefer_Vogelsang2005,Sun_etc2008}. The series HAR approach considered in this paper is recently introduced in the econometrics literature by \citet{Phillips2005}, \citet{Mueller2007}, and \citet{Sun2011,Sun2013} as an alternative to kernel-based methods, and further studied in various settings by \citet{Chen_Liao_Sun2014}, \citet{Lazarus_etc2018}, and \citet{Rho_Vogelsang2021}. Extending existing HAR methods to the two-sample setting is nontrivial, because the groups often exhibit heterogeneous LRVs, in which case the standard HAR inference is not directly applicable. 

As in classical two-sample inference, we consider two cases. First, we study the case in which the two time series have equal LRVs. In this setting, we adopt the fixed-$K$ asymptotics, under which the numbers of basis functions, $K_1$ and $K_2$, are held fixed as $T_1$ and $T_2$ grow. We show that our $t$- and scaled Wald statistics, which exploit the equality of LRVs, converge to the Student's $t$- and $F$-distributions, respectively, under $H_0$. 

We next consider the empirically more relevant case in which the two time series have unequal LRVs. In this setting, we introduce series HAR $t$- and Wald statistics that accommodate heterogeneous LRVs. Under fixed-$K$ asymptotics, however, these statistics are no longer asymptotically pivotal, so asymptotically valid Student's $t$- and $F$-tests are unavailable. 
To address this, we pursue two approaches.  

First, in the spirit of Welch's $t$-test, we propose a Student’s $t$-approximation---and its multivariate extension, an $F$-approximation---with degrees of freedom adjusted to account for unequal LRVs. This approach is simple to implement---critical values are obtained from the standard Student's $t$- and $F$-distributions---and it substantially improves finite-sample accuracy relative to the normal and chi-square approximations. It is also well grounded in the literature. See, for example, \citet{Bell_McCaffrey2002}, \citet{Imbens_Kolesar2016}, \citet{Donald_Lang2007}, and \citet{Sun2014}.

Second, we develop a series HAR wild bootstrap (SHAR-WB) and propose SHAR-WB based two-sample inference---and its multivariate extension---that is robust to both serial dependence and unequal LRVs. SHAR-WB builds on the wild bootstrap of \citet{Wu1986} and the series HAR framework. Specifically, we introduce serial dependence into the external random variables used to perturb the residuals, thereby allowing the bootstrap samples to replicate the dependence structure of the original data. Unlike block bootstrap methods, SHAR-WB does not require resampling blocks of observations. Instead, it generates dependent external random variables using orthogonal bootstrap basis functions so that the LRV of bootstrap samples mirrors the formulation of the series HAR estimator. We establish the asymptotic validity by showing that the bootstrap distributions of the bootstrap $t$- and Wald statistics are consistent for those of the original series HAR counterparts under increasing-$K$ asymptotics. SHAR-WB aligns well with our series HAR-based two-sample tests and delivers superior finite-sample performance. It is of independent interest for other testing problems involving dependent data.

The proposed SHAR-WB is closely related to the dependent wild bootstrap (DWB) of \citet{Shao2010}, which generates external random variables whose covariances are kernel weighted across time lags. The resulting bootstrap LRV yields a kernel-based LRV estimator. \citet{Conley_etc2023} extend this idea to the spatial setting. See also related work by \citet{Leucht_Neumann2013}, \citet{Djogbenou_etc2015}, and \citet{Hounyo2023}. Further discussion of the relationship between our SHAR-WB and DWB is provided in Sections \ref{Section: SHAR-WB} and \ref{Sec:appendix}.

Beyond the scalar mean, we develop the framework for \emph{joint} inference on a vector of means, so that equality can be tested across several outcomes---or several contrasts---simultaneously. This multivariate step is what many applications require: inference on the difference-in-differences (DID) treatment effect is itself a special case of our two-sample framework, since the DID estimand is a contrast of group means in suitably differenced series, and testing a vector of such effects---across multiple outcomes or post-treatment horizons---is inherently multivariate. For a $p$-dimensional null we form a series HAR Wald statistic, $W_{1,\mathrm{HAR}}$, which converges to a $\chi^2_p$ distribution under increasing-$K$ asymptotics (the HAR-$\chi^2$ test). Because this increasing-$K$ approximation over-rejects under persistent dependence, we propose a fixed-$K$ $F$-approximation (HAR-$F$)---the multivariate counterpart of our $t$-approximation, with degrees of freedom adjusted for long-run-variance heterogeneity across coordinates---together with a multivariate version of the series HAR bootstrap. In the scalar case the fixed-$K$ $F$-approximation reduces exactly to our $t$-approximation.


Our paper is closely related to \citet{Ibragimov_Mueller2016}, who develop a $t$-test for the comparison of a scalar parameter across two groups. Their approach partitions the observations in each group into a small number of clusters and estimates the parameter separately within each cluster. The $t$-statistic is then constructed from the resulting cluster-level estimates. They employ adjusted degrees of freedom following \citet{Bakirov1998} and show its asymptotic validity under broad forms of heterogeneity and dependence. Their approach is applicable to a broad range of settings including time-series, spatial, and panel data and is particularly attractive when observations naturally exhibit a grouped structure, whereas our paper focuses on time-series inference. By contrast, our framework accommodates multivariate parameters, whereas their procedure is restricted to scalar parameters. Furthermore, our SHAR-WB procedure can provide improved finite-sample performance, particularly in the presence of strong serial dependence or relatively small sample sizes.  A comparison of the finite-sample performance is presented in Section \ref{Section:Simulation}.

We illustrate the relevance of our proposed methodology through two empirical applications. The first revisits a study previously analyzed in the literature that examines the impact of working from home (WFH) on employee performance, testing both the individual performance outcomes and, jointly, the full vector of outcomes. The second considers two key U.S. macroeconomic time series---the civilian unemployment rate and CPI inflation---which share a well-documented structural break around the 2008 financial crisis, and tests whether their bivariate mean vector shifts across the break.

Using the proposed fixed-$K$ HAR-$F$ test and the SHAR-WB in both scalar and multivariate forms, we find no statistically meaningful difference in average performance between the WFH and office-based groups---whether the outcomes are examined one at a time or jointly through the multivariate SHAR-WB test---once sampling variation, heteroskedasticity, and autocorrelation are properly accounted for. Likewise, at the $5\%$ significance level, we find no statistically meaningful difference in the means across the pre- and post-break periods for the two U.S. macroeconomic series---the Unemployment Rate and CPI Inflation---either component by component or in the joint bivariate test. In both applications, several of these differences appear strongly significant under the classical and Welch tests that ignore serial dependence, underscoring the practical value of the proposed robust procedures.

The remainder of the paper is organized as follows. Section \ref{Section: Setup} outlines the setup and reviews the conventional two-sample $t$-tests, including Welch's version. Section \ref{Section: HAR-t-tests} introduces the series-based HAR two-sample tests and studies their asymptotic properties. Section \ref{Section: SHAR-WB} develops the series HAR wild bootstrap. Section \ref{Section:Simulation} evaluates the finite-sample performance of the proposed tests relative to existing methods. Section \ref{Section:Empirical} presents two empirical applications, and Section \ref{Section:Conclusion} concludes. Additional discussion, technical derivations, and additional simulation and empirical results are provided in Section \ref{Sec:appendix}.

\section{Setup and Classical \texorpdfstring{$t$}{t}-Tests} \label{Section: Setup}
Suppose that $\{Y_{1t}\}$ and $\{Y_{2t}\}$ are $p \times 1$ time-series processes with 
\begin{equation*}
Y_{jt} \sim  \left( \mu _{j},\Sigma
_{j}\right) ,\quad  t=1,...,T_{j},\ j=1,2,
\end{equation*}%
where the two processes are independent of each other. Let $\Delta=\mu_1-\mu_2$. We are interested in testing the difference of population means.
\begin{equation}
H_{0}:\Delta=\Delta_{0} .  \label{H0: Two Sample}
\end{equation}
$\{Y_{1t} \}$ and $\{Y_{2t}\} $ can be (i) pre- and
post-structural break data with known break date, and (ii) time series of
two different units (e.g., treatment unit and control unit). We can also consider the case of two groups, each consisting of multiple units over time (panel or repeated cross-section structure). In this case, $\{Y_{1t}\}$ and $\{Y_{2t}\}$ represent the time series of cross-sectional averages for
the two groups. 


Many important econometric procedures can also be formulated as two-sample inference problems. Consider the difference-in-differences (DID) regression:
    \[
        \tilde{Y}_{at} = \beta_0 + \beta_1 Treat_{a} + \beta_2 Post_t + \beta_3 (Treat_{a}\times Post_t) + \tilde{u}_{at},   \quad a=0,1,     
    \]
where $Treat_{a}$ and $Post_t$ are binary variables for the treatment group and the post-treatment period, respectively. Let $Y_{1t}$ denote the post-treatment difference between the treatment and control groups and $Y_{2t}$ denote the corresponding pre-treatment difference. Then, 
\[
\beta_3 = E[Y_{1t}]-E[Y_{2t}] = \Delta, 
\]
which implies that inference on the DID treatment effect is a special case of our two-sample inference framework. 

Our framework is also closely related to the Diebold--Mariano (DM) test for predictive accuracy \citep{diebold1995comparing}.\footnote{The Diebold--Mariano test has become one of the most widely used procedures for predictive accuracy testing in econometrics and statistics. As of July 15, 2026, the original paper has received more than 13,700 citations according to Google Scholar.} The DM procedure tests the null hypothesis of equal predictive ability by constructing a loss differential process,
\[
d_t=\ell(e_{1t})-\ell(e_{2t}),
\]
where $\ell(\cdot)$ is a loss function, and conducting inference on the mean of \(\{d_t\}_{t=1}^T\) under serial dependence using a kernel-based heteroskedasticity and autocorrelation consistent (HAC) estimator of the long-run variance. From our perspective, the DM test can be viewed as a special case of two-sample mean inference under serial dependence in which the two samples are paired and collapsed into a single loss-differential process with equal sample sizes. In contrast, our framework directly addresses inference on 
\[
\Delta = E[Y_{1t}]-E[Y_{2t}],
\]
allowing for heterogeneous dependence structures across samples and unequal sample sizes (\(T_1 \neq T_2\)). Moreover, our approach employs orthonormal-basis projection standardization and fixed-$K$ asymptotic and bootstrap methods to provide more reliable inference in finite samples. Conversely, the DM test has an advantage when the two time series are paired, as it can accommodate cross-sectional dependence across units, whereas our approach requires independence across units.

When ${Y_{1t}}$ and ${Y_{2t}}$ are scalar ($p=1$), the classical two-sample $t$-test is the standard procedure for comparing the two population means in (\ref{H0: Two Sample}). Assuming that $\{Y_{1t}\}$ and $\{Y_{2t}\}$ are serially independent and normally
distributed with equal variances, i.e., $\Sigma _{1}=\Sigma _{2},$
the $t$-statistic is defined as
\begin{equation*}
t_{0}=\frac{\bar{Y}_{1}-\bar{Y}_{2} - \Delta_0}{s_{p}\sqrt{\frac{1}{T_{1}}+\frac{1}{%
T_{2}}}},
\end{equation*}%
where $\bar{Y}_{j}=\frac{1}{T_{j}}\sum_{t=1}^{T_{j}}Y_{jt} $ and $s_{p}=\sqrt{\frac{\left( T_{1}-1\right) s_{1}^{2}+\left(
T_{2}-1\right) s_{2}^{2}}{T_{1}+T_{2}-2}}$ with $s_{j}^{2}=\frac{1}{T_{j}-1}%
\sum_{t=1}^{T_{j}}\left( Y_{jt}-\bar{Y}_{j}\right) ^{2}.$ Under $H_{0},$ $%
t_{0}$ follows an exact Student's $t$-distribution with $%
T_{1}+T_{2}-2$ degrees of freedom.

If the variances are heterogeneous, i.e., $\Sigma _{1}\neq
\Sigma _{2}$, the original two-sample $t$-test described above is no longer valid. In such cases, Welch's $t$-test \citep{Welch1947} is commonly recommended, as it accounts for unequal variances by adjusting the degrees of freedom. The test is based on  
\begin{equation*}
t_{1}=\frac{\bar{Y}_{1}-\bar{Y}_{2} - \Delta_0}{\sqrt{s_{1}^2/T_{1}+s_{2}^2/
T_{2}}},
\end{equation*}%
and the critical values are obtained from a Student's $t$-distribution
with the degrees of freedom computed using the Welch-Satterthwaite equation:  
\begin{equation*}
\frac{\left( s_{1}^{2}/T_{1}+s_{2}^{2}/T_{2}\right) ^{2}}{%
\left( s_{1}^{2}/T_{1}\right) ^{2}/\left( T_{1}-1\right) +\left(
s_{2}^{2}/T_{2}\right) ^{2}/\left( T_{2}-1\right) }.
\end{equation*}

Both the classical two-sample $t$-test and Welch's two-sample $t$-test assume independent, normally distributed observations. While the impact of violating normality becomes negligible in large samples, the independence assumption remains crucial, and these tests may severely over-reject the null in the presence of serial dependence. This is a serious limitation of the $t$-tests with time-series data, since serial dependence is the norm rather than the exception.

\section{Series HAR Two-Sample Tests} \label{Section: HAR-t-tests}
In this section, we develop series HAR inference methods for comparing two-sample means with time-series data. Under some regularity conditions and weak dependence, we have
\[
    \sqrt{T_j} \left( \bar{Y}_j - \mu_j \right) \rightarrow^d N(0,\Omega_j), \ \ j=1,2,
\]
as $T_j \rightarrow \infty$, where $\Omega_j = \lim_{T_j\rightarrow \infty} Var\left(\sqrt{T_j} \left( \bar{Y}_j - \mu_j \right)\right)$ denotes the long-run variance (LRV) of $\{Y_{jt}\}$. This asymptotic normality provides the  basis for testing (\ref{H0: Two Sample}), which requires the estimation of $\Omega_1$ and $\Omega_2$. 

We adopt the series HAR approach. Let $\{\phi_\ell(\cdot)\}_{\ell=1}^K$ denote a set of orthonormal basis functions on $[0,1]$. For each $\ell = 1,...,K_j$, project the residuals and form the outer product of the projection coefficient:
\[
\hat{z}_{j,\ell} := \frac{1}{\sqrt{T_j}}\sum_{t=1}^{T_j} \phi_\ell\left(\frac{t}{T_j}\right) \hat{u}_{jt}, \ \ 
\hat{\Omega}_{j,\ell} := \hat{z}_{j,\ell} \hat{z}_{j,\ell}^\prime,
\]
where $\hat{u}_{jt} ={Y}_{jt}-\bar{Y}_j $.
Under mild conditions, each $\hat{\Omega}_{j,\ell}$ is asymptotically unbiased but inconsistent for $\Omega_j$. The series HAR variance estimator is defined as the simple average: 
\[
\hat{\Omega}_j = \frac{1}{K_j} \sum_{\ell=1}^{K_j} \hat{\Omega}_{j,\ell}.
\]
Because $\{\hat{\Omega}_{j,\ell} \}_{\ell=1}^{K_j}$ are asymptotically unbiased, depending on how $K_j$ behaves as $T_j \rightarrow \infty$, two asymptotic frameworks, fixed-$K$ asymptotics and increasing-$K$ asymptotics, are available to characterize the asymptotics of $\hat{\Omega}_j$ and the associated test statistics. Under increasing-$K$ asymptotics, where $K_j \rightarrow \infty$ as $T_j \rightarrow \infty$ with $K_j/T_j \rightarrow 0$, $\hat{\Omega}_j$ is consistent for $\Omega_j$. Under fixed-$K$ asymptotics, where $K_j$ is held fixed as $T_j \rightarrow \infty$, $\hat{\Omega}_j$ converges in distribution to a scaled Wishart random matrix proportional to $\Omega_j$. \citet{Sun2013} shows that fixed-$K$ asymptotic tests provide a high-order refinement of the null rejection probability in Wald test settings, comparable to kernel-based fixed-$b$ tests.

We make the following assumptions on the basis functions and data to develop the series HAR two-sample tests under fixed-$K$ asymptotics. 

\begin{assumption}
\label{Assumption 1} $\{\phi_{\ell }\left( \cdot \right)\}_{\ell=1}^K$ is a sequence of piecewise monotonic, continuously
differentiable and orthonormal basis functions on $L_{2}\left[ 0,1\right]$ satisfying $\int_0^1 \phi_\ell(x) dx =0$. 
\end{assumption}
Assumption \ref{Assumption 1} is standard, and sine and cosine
trigonometric polynomials are commonly used in the literature. The mean zero condition, $\int_0^1 \phi_\ell(x) dx =0$, yields
\[
\frac{1}{\sqrt{T_j}}\sum_{t=1}^{T_j} \phi_\ell\left(\frac{t}{T_j}\right) \hat{u}_{jt} = \frac{1}{\sqrt{T_j}}\sum_{t=1}^{T_j} \phi_\ell\left(\frac{t}{T_j}\right) u_{jt} + o_P(1),
\]
where the effect of the estimation error in $\bar{Y}_j$ is asymptotically negligible. In this paper, we employ the following set of basis functions for the asymptotic analysis and for simulations and empirical studies.
\begin{equation}
\label{Basis function}
\phi_{\ell}\left( x\right) =\left\{ 
\begin{array}{c}
\sqrt{2}\cos \left( 2\pi \ell x\right) , \\ 
\sqrt{2}\sin \left( 2\pi \ell x\right) ,%
\end{array}%
\begin{array}{l}
\text{if } \ell \text{ is odd} \\ 
\text{if } \ell \text{ is even}%
\end{array}%
,\ \ell=1,...,K.\right.
\end{equation}
\begin{assumption}\label{Assumption 2}Let $u_t = \left(u_{1t}^\prime,u_{2t}^\prime \right)^\prime$. The functional CLT holds:%
\[
\frac{1}{\sqrt{T}}\sum_{t=1}^{\lfloor rT\rfloor } u_{t} \rightarrow ^{d} 
\begin{pmatrix}
\Omega_1^{1/2} & 0\\
0 & \Omega_2^{1/2}
\end{pmatrix} W_2(r), 
\]
where $W_2(r)=\left(W_{21}^\prime(r),W_{22}^\prime(r) \right)^\prime$ is the $2p$-dimensional standard Brownian motion.
\end{assumption}
The functional CLT in Assumption \ref{Assumption 2} is a high level condition that is commonly used to facilitate the HAR inference asymptotic theory  \citep[e.g.,][]{Kiefer_Vogelsang2002,Kiefer_Vogelsang2005}. Let $\hat{u}_t = \left(\hat{u}_{1t}^\prime,\hat{u}_{2t}^\prime \right)^\prime$ and $\phi_0(x) = 1$ for $x \in [0,1]$. For each $\ell=0,1,...,K$,
\begin{eqnarray*}
\frac{1}{\sqrt{T}}\sum_{t=1}^{T} \phi_{\ell}\!\left(\frac{t}{T}\right)\hat{u}_{t} &=&\frac{1}{\sqrt{T}}\sum_{t=1}^{T} \phi_{\ell}\!\left(\frac{t}{T}\right)u_{t} + o_P(1) \\
&\rightarrow^d& \begin{pmatrix}
\Omega_1^{1/2} & 0\\
0 & \Omega_2^{1/2}
\end{pmatrix}  \int_0^1 \phi_\ell(r) dW(r)
\end{eqnarray*}
under Assumptions \ref{Assumption 1} and \ref{Assumption 2}, which, by the continuous mapping theorem, implies
\begin{eqnarray} \label{eqn: chi2_1}
    \hat{\Omega}_{j,\ell} &=& \left( \frac{1}{\sqrt{T_j}}\sum_{t=1}^{T_j} \phi_{\ell}\!\left(\frac{t}{T_j}\right)\hat{u}_{jt} \right) \left( \frac{1}{\sqrt{T_j}}\sum_{t=1}^{T_j} \phi_{\ell}\!\left(\frac{t}{T_j}\right)\hat{u}_{jt} \right)^\prime \notag \\
    &\rightarrow^d& \Omega_j^{1/2} \mathcal{W}(I_{p},1) \Omega_j^{1/2},
\end{eqnarray}
where $\mathcal{W}(I_{p},1)$ denotes a Wishart distribution with scale matrix $I_{p}$ and one degree of freedom. In the scalar case ($p=1$), $W_{p}(I_{p},1)$ reduces to the chi-square distribution with one degree of freedom, so 
\[
    \hat{\Omega}_{j,\ell} \rightarrow^d \Omega_j \chi^2(1).
\]
In addition, by the orthonormality and mean zero properties of the basis functions, 
\begin{eqnarray} \label{eqn: independence}
Cov\left(\int_0^1 \phi_k(r) dW_{2j}(r),\int_0^1 \phi_\ell(r) dW_{2j}(r)\right) = 0    
\end{eqnarray}
for $k \neq \ell$ and $k,\ell = 0,1,...,K$. 

In what follows, we consider two cases---depending on whether the LRVs are equal---and develop corresponding two-sample tests that are robust to serial dependence. 

\subsection{Case 1: Equal LRVs} \label{Subsection: Equal LRVs}

We first consider the simple case where $\{Y_{1t}\}$ and $\{Y_{2t}\}$ have equal LRVs. Specifically, we set $\Omega:=\Omega_1 = \Omega_2$. In this setting, we use the pooled LRV estimator to define the two-sample Wald statistic 
\[
    W_{0,\mathrm{HAR}} =  \left(\bar{Y}_{1}-\bar{Y}_{2} - \Delta_0 \right)^\prime \left( \hat{\Omega}_{\text{pool}} \left(\frac{1}{T_{1}}+\frac{1}{T_{2}}\right)\right)^{-1} \left(\bar{Y}_{1}-\bar{Y}_{2} - \Delta_0 \right),
 \]
where $\hat{\Omega}_{\text{pool}} = \frac{K_1 \hat{\Omega}_1 + K_2 \hat{\Omega}_2}{K_1 + K_2}$. By (\ref{eqn: chi2_1}) and (\ref{eqn: independence}),  the equality of the LRVs implies
\[
 K_1 \hat{\Omega}_1 + K_2 \hat{\Omega}_2 \rightarrow^d \mathcal{W}(\Omega,K_1 + K_2) = \Omega^{1/2} \mathcal{W}(I_{p},K_1 + K_2) \Omega^{1/2} 
\]
and the numerator and denominator of $W_{0,\mathrm{HAR}}$ are asymptotically independent. 

Therefore, if the null hypothesis in (\ref{H0: Two Sample}) is true,
\[
    W_{0,\mathrm{HAR}} \rightarrow^d T^2(p,K_1+ K_2),
\]
where $T^2(p,K_1+K_2)$ denotes Hotelling's $T^2$ distribution. Furthermore, for $K_1 + K_2 \geq p$,
\[
     m_{K,p} T^2(p,K_1+ K_2) \sim F(p,K_1+ K_2 -p+1), \quad m_{K,p} = \frac{K_1+K_2 - p+ 1}{p(K_1+K_2)},
\]
where  $m_{K,p}$ is the corresponding scaling factor, and $F(p,K_1+ K_2 -p+1)$ denotes the $F$-distribution with $p$ and $K_1+K_2-p+1$ degrees of freedom. Thus, we obtain 
\[
   m_{K,p} W_{0,\mathrm{HAR}} \rightarrow^{d} F(p,K_1 + K_2-p+1).
\]
 When ${Y_{1t}}$ and ${Y_{2t}}$ are scalar ($p=1$), the proposed framework reduces to the two-sample $t$-test, which implies that, under $H_0$,
\begin{equation*}
t_{0,\mathrm{HAR}}=\frac{\bar{Y}_{1}-\bar{Y}_{2} - \Delta_0}{\hat{\Omega}_{\text{pool}}^{1/2}\sqrt{\frac{1}{T_{1}}+\frac{1}{%
T_{2}}}}  \rightarrow^d t(K_1+K_2).
\end{equation*}%

The result is summarized below.
\begin{theorem}\label{Theorem 1} Suppose that Assumptions \ref{Assumption 1} and \ref{Assumption 2} hold, and that $\Omega_1 = \Omega_2$. Then, under $H_0$,  
\[
     m_{K,p} W_{0,\mathrm{HAR}} \rightarrow^d F(p,K_1 + K_2-p+1) \quad \text{and} \quad t_{0,\mathrm{HAR}} \rightarrow^d t(K_1 + K_2),
\] 
as $T_1,T_2 \rightarrow \infty$, for fixed $K_1$ and $K_2$.
\end{theorem}

The proof is provided in Section \ref{Sec:appendix}. It is worth noting that one can employ a kernel-based fixed-$b$ approach in this setting. A practical advantage of our series-based HAR test over the kernel fixed-$b$ method is that the critical values are taken from the standard Student's $t$- and \(F\)-distributions, whereas the limiting distribution of the latter is nonstandard and its critical values should be simulated.

\subsection{Case 2: Unequal LRVs}
We now consider the empirically more relevant case in which the LRVs are heterogeneous, i.e., $\Omega_1 \neq \Omega_2$.  Because the pooled LRV estimator $\hat{\Omega}_{\text{pool}}$ is no longer appropriate, we instead consider the Wald statistic
\[
W_{1,\mathrm{HAR}} =  \left(\bar{Y}_{1}-\bar{Y}_{2} - \Delta_0 \right)^\prime \left( \frac{\hat{\Omega}_1}{T_1} + \frac{\hat{\Omega}_2}{T_2} \right)^{-1} \left(\bar{Y}_{1}-\bar{Y}_{2} - \Delta_0 \right),
\]
which accommodates heterogeneous LRVs. The corresponding $t$-statistic is
\[
t_{1,\mathrm{HAR}} = \frac{\bar{Y}_{1}-\bar{Y}_{2} - \Delta_0}{\sqrt{\frac{\hat{\Omega}_1}{T_1} + \frac{\hat{\Omega}_2}{T_2}}}.   
\]
If $\Omega_1 \neq \Omega_2$, neither $W_{1,\mathrm{HAR}}$ nor $t_{1,\mathrm{HAR}}$ is asymptotically pivotal under fixed-$K$ asymptotics, so asymptotically valid $F$- and $t$-tests are unavailable in this setting. Therefore, we first establish their asymptotics under increasing-$K$ asymptotics.

We assume $\{ u_{jt} \}$ is fourth order stationary and introduce the following assumption that controls the degree of serial dependence of $\{u_{jt}\}$. Let  $\Gamma_j(h) = E\left(u_{jt} u_{jt+h}^\prime \right)$ and define the fourth order cumulant of $\left(u_{jt}^{(a)},u_{jt+h_1}^{(b)},u_{jt+h_2}^{(c)},u_{jt+h_3}^{(d)}\right)$ by
\begin{eqnarray*}
    &&\kappa_{j,abcd}(0,h_1,h_2,h_3) \\
    &=& E\left( u_{jt}^{(a)} - Eu_{jt}^{(a)} \right)\left( u_{jt+h_1}^{(b)} - Eu_{jt+h_1}^{(b)}\right)\left( u_{jt+h_2}^{(c)} - Eu_{jt+h_2}^{(c)}\right)\left( u_{jt+h_3}^{(d)} - Eu_{jt+h_3}^{(d)} \right) \\
    &&-E\left( \ddot{u}_{jt}^{(a)} - E\ddot{u}_{jt}^{(a)} \right)\left( \ddot{u}_{jt+h_1}^{(b)} - E\ddot{u}_{jt+h_1}^{(b)}\right)\left( \ddot{u}_{jt+h_2}^{(c)} - E\ddot{u}_{jt+h_2}^{(c)}\right)\left( \ddot{u}_{jt+h_3}^{(d)} - E\ddot{u}_{jt+h_3}^{(d)} \right),
\end{eqnarray*}
where $\{\ddot{u}_{jt}\}$ is a Gaussian sequence with the same mean and covariance structure as $\{u_{jt}\}.$
\begin{assumption}\label{Assumption3: cumulant}
    Let $\{ u_{jt} \}$ for $j=1,2,$ be mean zero, fourth order stationary processes with $\sum_{h=-\infty}^{\infty}|h|^3|\Gamma_j(h)| < \infty$ and $\sum_{h_1=-\infty}^{\infty}\sum_{h_2=-\infty}^{\infty}\sum_{h_3=-\infty}^{\infty} \left| \kappa_{j,abcd}(0,h_1,h_2,h_3) \right| < \infty$ for all $a,b,c,d \leq p$.  
\end{assumption}
The cumulant condition in Assumption \ref{Assumption3: cumulant} is standard in the time-series literature. It is trivial in the Gaussian case, since the fourth-order cumulant is zero. This condition also holds for linear processes under coefficient summability and suitable moment conditions \citep[e.g.,][]{Kim_Sun2011}. Primitive sufficient conditions based on $\alpha$-mixing and a moment condition are provided in \citet[Lemma 1]{Andrews1991}.

\begin{theorem} \label{Theorem: Asymptotic Normality}
    Suppose Assumptions \ref{Assumption 1}-\ref{Assumption3: cumulant} hold. Then, under $H_0$,
    \[
    W_{1,\mathrm{HAR}} \rightarrow^d \chi^2(p) \quad \text{and} \quad t_{1,\mathrm{HAR}} \rightarrow^d N(0,1),
    \]
    as $T_j \rightarrow \infty$ and $K_j \rightarrow \infty$ with $K_j/T_j \rightarrow 0$ for $j=1,2.$
\end{theorem}
The proof is provided in Section \ref{Sec:appendix}. Theorem \ref{Theorem: Asymptotic Normality} implies the asymptotic validity of the chi-square and normal tests based on $W_{1,\mathrm{HAR}}$ and $t_{1,\mathrm{HAR}}$, respectively, under increasing-$K$ asymptotics. However, it is well documented that such increasing-$K$ asymptotic tests based on the consistency of HAR estimators can exhibit severe size distortions in finite samples, particularly under strong serial dependence. To address this issue, we consider two approaches. First, in the spirit of Welch's $t$-test, we approximate the limiting distribution of the scaled Wald statistic by an $F$-distribution, and that of $t_{1,\mathrm{HAR}}$ by a Student's $t$-distribution, with adjusted degrees of freedom to account for the unequal LRVs. Second, we develop a series HAR wild bootstrap (SHAR-WB) procedure together with the corresponding two-sample tests and establish its asymptotic validity under increasing-$K$ asymptotics. See Section \ref{Section: SHAR-WB} for the second approach. The remainder of this section develops the first approach. 
\paragraph{Student's $t$- and $F$-approximations} As discussed in Section \ref{Section: Setup}, under independence, one popular way to handle unequal variances is Welch's $t$-test, which approximates the sampling distribution of the two-sample $t$-statistic by a Student's $t$-distribution with adjusted degrees of freedom based on the Welch-Satterthwaite equation. \citet{Bell_McCaffrey2002} apply Welch's approach to linear regression, and \citet{Imbens_Kolesar2016} further extend it to clustering settings. Motivated by these developments, we extend the idea of distributional approximation to the present setting by approximating the fixed-$K$ limiting distribution of $t_{1,\mathrm{HAR}}$ with a Student's $t$-distribution, and that of its multivariate counterpart, scaled Wald statistic, with an $F$-distribution, using Welch-type adjusted degrees of freedom.

Define $\hat{\Omega}_{1,\mathrm{HAR}} := \sqrt{T_2/T_1}\,\hat{\Omega}_1 + \sqrt{T_1/T_2}\,\hat{\Omega}_2$. Under $H_0$ and fixed-$K$ asymptotics the limiting distribution of $W_{1,\mathrm{HAR}}$ is non-pivotal, and we approximate that of $\tfrac{K-p+1}{pK}\,W_{1,\mathrm{HAR}}$ by an $F(p,\,K-p+1)$ distribution. We choose the degrees of freedom $K=K_{\mathrm{adf}}$ by matching the first two asymptotic moments of $\mathrm{tr}(\hat{\Omega}_{1,\mathrm{HAR}})$ to those of the reference scaled-Wishart form $\mathrm{tr}\!\big(V^{1/2}\,\mathcal{W}(I_p,K_{\mathrm{adf}})/K_{\mathrm{adf}}\,V^{1/2}\big)$, where $V=\rho^{1/2}\Omega_1+\rho^{-1/2}\Omega_2$ and $\rho=\lim T_2/T_1$; the derivation is given in Section \ref{Sec:appendix}. The means coincide for any $K_{\mathrm{adf}}$, and matching the variances yields
\begin{eqnarray} \label{eqn: infeasible adf}
K_{\text{adf}} = \frac{\text{tr}\left[\left(\rho^{1/2} \Omega_1 + \rho^{-1/2} \Omega_2 \right)^2\right]}{  \rho \text{tr}(\Omega_1^2)/K_1 +  \rho^{-1} \text{tr}(\Omega_2^2)/K_2}.
\end{eqnarray}
(\ref{eqn: infeasible adf}) highlights two informative cases. First, when $\text{tr}(\Omega_1) \gg \text{tr}(\Omega_2)$ or $T_1 \ll T_2$, $K_{\text{adf}}$ approaches $K_1$. Intuitively, the adjusted degrees of freedom are driven by the side with the larger LRV. Second, in the balanced case with $\Omega_1 = \Omega_2$ and $T_1 = T_2$, we obtain $K_{\text{adf}} = \frac{4 K_1 K_2}{K_1+K_2}$, which is no greater than $K_1 + K_2$, the degrees of freedom in the case of equal LRVs. The equality holds only when $K_1 = K_2$, in which case $W_{0,\mathrm{HAR}} = W_{1,\mathrm{HAR}}$.

Let $\tilde{\Omega}_j$ denote a consistent estimator of $\Omega_j$. A feasible version is
\begin{eqnarray} \label{eqn: feasible adf}
\hat {K}_{\text{adf}} = \frac{\text{tr}\left[\left(\rho_T^{1/2}\tilde{\Omega}_1 + \rho_T^{-1/2} \tilde{\Omega}_2 \right)^2\right]}{ \rho_T \tilde{\Omega}_1^2/K_1 +  \rho_T^{-1} \tilde{\Omega}_2^2/K_2}    
\end{eqnarray}
where $\rho_T = T_2/T_1.$ In practice, one may set $\tilde{\Omega}_j=\hat{\Omega}_j,$ and we adopt this choice in our simulations and empirical studies. Alternatively, to better reflect the increasing-$K$ asymptotics for consistency, $\tilde{\Omega}_1$ and $\tilde{\Omega}_2$ can be constructed using larger numbers of basis functions.

If ${Y_{1t}}$ and ${Y_{2t}}$ are scalar ($p=1$), the proposed framework reduces to the univariate case, and the fixed-$K$ limiting distribution of $t_{1,\mathrm{HAR}}$ can be approximated by the Student's $t$-distribution with $K_{\text{adf}}$ degrees of freedom.

\section{Series HAR Wild Bootstrap} \label{Section: SHAR-WB}
\subsection{Bootstrap Method}

In this section, we propose a novel series HAR wild bootstrap (SHAR-WB) test for two-sample comparison with time-series data. The bootstrap data generating process is described as follows. Define
\begin{eqnarray} \label{eqn: bootstrap DGP 1}
Y_{jt}^\ast \;=\; \mu_{0j}^\ast + u_{jt}^\ast,
\qquad
\mu_{01}^\ast \;=\; \frac{T_1 \bar Y_1 + T_2 \bar Y_2}{T_1+T_2} + \frac{\Delta_0}{2}, \qquad \mu_{02}^\ast = \mu_{01}^\ast - \Delta_0.    
\end{eqnarray}
Generate the bootstrap errors from 
\begin{eqnarray} \label{eqn: bootstrap DGP 2}
u_{jt}^\ast \;=\; \hat u_{jt}\,\eta_{jt},
\qquad
\hat u_{jt}=Y_{jt}-\bar{Y}_j,    
\end{eqnarray}
where $\eta_{jt}$ is an external random variable, drawn independently of the data. Crucially, $\mu_{01}^\ast$ and $\mu_{02}^\ast$ are used to generate both series $\{Y_{1t}^\ast\}$ and $\{Y_{2t}^\ast\}$ in (\ref{eqn: bootstrap DGP 1}), thereby imposing the null of $ E^\ast Y_1^\ast - E^\ast Y_2^\ast = \Delta_0$ in the bootstrap world. 

In the standard wild bootstrap, $\eta_{jt}$ is generated in an i.i.d. fashion such that $E^\ast \left(\eta_{jt}\right) = 0$ and $Var^\ast \left( \eta_{jt} \right) = 1$. Consequently, conditional on the data, the bootstrap errors $u_{jt}^\ast$ are independently distributed with mean zero and variance $\hat{u}_{jt} \hat{u}_{jt}^\prime$. However, this approach is not appropriate in our time-series setting, because although it replicates the heteroskedasticity of the original data, it fails to capture their serial dependence.

To address this issue, we propose an SHAR-WB procedure that mimics the serial dependence in the original data nonparametrically. We generate $\{\eta_{jt}\}_{t=1}^{T_j}$ for $j=1,2$ that satisfy the following condition.

\begin{assumption}\label{assumption4: external rv}
The external random variables, $\{\eta_{jt}\}_{t=1}^{T_j}$ for $j=1,2$, are independent of data with $E^\ast\left(\eta_{jt}\right) = 0$, $Var^\ast\left(\eta_{jt}\right) = 1$, and 
\[
Cov^\ast\left(\eta_{jt},\eta_{js}\right) = \frac{1}{K_j^\ast} \sum_{\ell=1}^{ 2 K_j^\ast} \psi_{\ell}\left( \frac{t}{T_j} \right)\psi_{\ell}\left( \frac{s}{T_j} \right),
\]
where $\{\psi_\ell\left(\cdot\right)\}_{\ell=1}^{2K_j^\ast}$ is a sequence of piecewise monotonic, continuous differentiable and orthogonal basis functions on $L_2[0,1]$ satisfying $\int_0^1 \psi_\ell^2(x)dx = \frac{1}{2}$ and $\int_0^1 \psi_\ell(x)dx = 0$.
\end{assumption}

In this paper, we employ
\begin{eqnarray} \label{eqn: Basis function 1}
\psi_{2\ell-1}= \cos(2\pi \ell x):=\psi_{1,\ell}(x) \qquad \psi_{2\ell} = \sin(2\pi \ell x):=\psi_{2,\ell}(x)
\end{eqnarray}
for $\ell = 1,...,K^\ast$. We then generate $\eta_{jt}$ from
\begin{eqnarray} \label{eqn: eta}
\eta_{jt} = \frac{1}{\sqrt{K_j^\ast}} \sum_{\ell=1}^{K_j^\ast} \left[\psi_{1,\ell}\left( \frac{t}{T_j} \right) v_{j,1\ell} + \psi_{2,\ell}\left( \frac{t}{T_j} \right) v_{j,2\ell} \right],    
\end{eqnarray}
where $v_{j,1\ell}$ and $v_{j,2\ell}$ are i.i.d. with mean 0 and variance 1, so that $E^\ast(\eta_{jt})=0$ and, using $\cos a \cos b + \sin a \sin b = \cos(a-b)$,
\begin{eqnarray} \label{eqn: cov of eta}
    Cov^\ast (\eta_{jt},\eta_{js}) =\frac{1}{K_j^\ast} \sum_{\ell=1}^{K_{j}^\ast} \cos\left(2\pi \ell \left(\frac{t-s}{T_j}\right)\right),
\end{eqnarray}
which gives $Var^\ast \left( \eta_{jt} \right)=1$.
Hence, as in the standard wild bootstrap, our procedure generates bootstrap errors $u_{jt}^\ast$ with mean zero and variance $\hat{u}_{jt} \hat{u}_{jt}^\prime$, conditional on the data. However, unlike the standard wild bootstrap, where $u_{jt}^\ast$ are independently distributed conditional on the data, SHAR-WB generates bootstrap errors that are dependently distributed, preserving the serial correlation structure of the original series.

\begin{remark}
We use two sets of basis functions with $2K_j^\ast$ i.i.d.\ draws $\{v_{j,1\ell},v_{j,2\ell}\}$ to obtain exact unit variance for every $t$; a single orthonormal set with $K_j^\ast$ draws yields $Var^{\ast}(\eta_{jt})=1$ only asymptotically.
\end{remark}

As shown in Section \ref{Sec:appendix}, SHAR-WB with the basis functions in (\ref{eqn: Basis function 1}) is asymptotically equivalent to a special case of the dependent wild bootstrap (DWB) of \citet{Shao2010} with a Daniell (sinc) kernel and bandwidth $M_j=T_j/(2K_j^\ast)$.


Our SHAR-WB yields the bootstrap LRV
\begin{eqnarray} \label{Omega_boot}
    \hat{\Omega}_{\text{boot},T_j} &=& Var^\ast\!\left( \frac{1}{\sqrt{T_j}} \sum_{t=1}^{T_j} \hat{u}_{jt}\eta_{jt} \right) \notag\\
    &=& \frac{1}{T_j} \sum_{t=1}^{T_j} \sum_{s=1}^{T_j}  \frac{1}{K_j^\ast} \sum_{\ell=1}^{K_{j}^\ast} \left[\psi_{1,\ell}\left( \frac{t}{T_j} \right)\psi_{1,\ell}\left( \frac{s}{T_j} \right) + \psi_{2,\ell}\left( \frac{t}{T_j} \right)\psi_{2,\ell}\left( \frac{s}{T_j} \right) \right]  \hat{u}_{jt} \hat{u}_{js}^\prime.
\end{eqnarray}

As shown in Lemma \ref{lem:consistency-omegahat} in Section \ref{Sec:appendix}, the SHAR-WB long-run variance, $\hat{\Omega}_{\text{boot},T_j}$, is a consistent estimator of the long-run variance, i.e., $\hat{\Omega}_{\text{boot},T_j} \rightarrow^P \Omega_j$ as $ K_j^\ast \rightarrow \infty$ as $T_j \rightarrow \infty$ such that $K_j^\ast/T_j \rightarrow 0$. 

\begin{theorem}\label{Theorem: Bootstrap consistency}
Suppose Assumptions \ref{Assumption 2}-\ref{assumption4: external rv} hold and $E^\ast v_{j,r\ell}^4 < C$ for $r=1,2$ and $C<\infty$. If $ K_j^\ast \rightarrow \infty$ as $T_j \rightarrow \infty$ such that $K_j^\ast/T_j \rightarrow 0$ for $j=1,2$, then, 
\[
\Omega_{j}^{-1/2} \sqrt{T_j}  \left( \bar{Y}_j^\ast - \mu_{0j}^\ast \right) \rightarrow^{d^\ast} N(0,I_p),
\]
in probability.
\end{theorem}
 See Section \ref{Sec:appendix} for the definition of $\rightarrow^{d^\ast}$ in probability. The proof is also in Section \ref{Sec:appendix}. Since $\sqrt{T_j} \left( \bar{Y}_j - \mu_j \right) \rightarrow^d N(0,\Omega_j)$,
Theorem \ref{Theorem: Bootstrap consistency} implies bootstrap distributional consistency for $\sqrt{T_j}\left( \bar{Y}_j^\ast - \mu^\ast \right),$ which underpins the bootstrap two-sample test stated in Theorem \ref{Theorem: Bootstrap test} below. Because the bootstrap LRV is consistent under increasing-$K$ asymptotics, it suffices to establish the asymptotic normality of the bootstrap score. The proof of asymptotic normality differs substantially from that of \citet{Shao2010}, who assumes that the external random variables are $\ell_n^\ast$-dependent and uses a large-block-small-block argument. In contrast, using (\ref{eqn: eta}), we rewrite the bootstrap score as
\begin{eqnarray*}
\sqrt{T_j} \bar{u}_j^\ast = \frac{1}{\sqrt{T_j}}\sum_{t=1}^{T_j}\hat{u}_{jt}\eta_{jt}  = \frac{1}{\sqrt{2K_j^\ast}} \sum_{\ell=1}^{K_j^\ast} \sum_{r=1}^2\left[\frac{\sqrt{2}}  {\sqrt{T_j}}\sum_{t=1}^{T_j}\hat{u}_{jt} \psi_{r,\ell}\left( \frac{t}{T_j} \right)\right] v_{j,r\ell}. 
\end{eqnarray*}
Since $v_{j,r\ell} \sim^{iid}(0,1)$, the bootstrap score is a sum of independent, heterogeneous random vectors conditional on the original sample, and Lyapunov's CLT applies under increasing-$K$ asymptotics.

\subsection{Bootstrap Two-Sample Tests}

In this subsection, we propose the SHAR-WB procedure for testing $H_0: \mu_1 - \mu_2 = \Delta_0$, while allowing for unequal LRVs. The bootstrap Wald and $t$-statistics are defined as:
\[
W_{1,\mathrm{HAR}}^\ast = \left(\bar{Y}_1^\ast - \bar{Y}_2^\ast - \Delta_0\right)^\prime \left(\frac{\hat{\Omega}_1^\ast}{T_1} + \frac{\hat{\Omega}_2^\ast}{T_2}\right)^{-1} \left(\bar{Y}_1^\ast - \bar{Y}_2^\ast - \Delta_0 \right), 
\]
and 
\begin{equation*}
t_{1,\mathrm{HAR}}^\ast=\frac{\bar{Y}_{1}^\ast-\bar{Y}_{2}^\ast - \Delta_0}{\sqrt{\frac{\hat{\Omega}_1^\ast}{T_1} + \frac{\hat{\Omega}_2^\ast}{T_2}}},
\end{equation*}%
where 
\[
\hat{\Omega}_j^\ast = \frac{1}{K_j} \sum_{\ell=1}^{K_j} \hat{\Omega}_{j,\ell}^\ast, \quad 
\hat{\Omega}_{j,\ell}^\ast := \hat{z}_{j,\ell}^\ast \hat{z}_{j,\ell}^{\ast \prime} , \quad \hat{z}_{j,\ell}^\ast := \frac{1}{\sqrt{T_j}}\sum_{t=1}^{T_j} \phi_\ell\left(\frac{t}{T_j}\right) \hat{u}_{jt}^\ast,
\]
with $\hat{u}_{jt}^{\ast} ={Y}_{jt}^{\ast}-\bar{Y}_j^{\ast} $.

We use the same basis functions and the same numbers of basis functions $K_1$ and $K_2$ to construct the series HAR estimators in both the original and bootstrap statistics, ensuring that the bootstrap and original studentizations are matched.

\begin{theorem} \label{Theorem: Bootstrap test}
    Suppose Assumptions \ref{Assumption 1}-\ref{assumption4: external rv} hold and $E^\ast v_{j,r\ell}^4 < C$ for $r=1,2$ and $C<\infty$.  If $K_j \rightarrow \infty$ and $K_j^\ast \rightarrow \infty$ as $T_j \rightarrow \infty$ such that $K_j / T_j \rightarrow 0$ and $K_j^\ast / T_j \rightarrow 0$ for $j=1,2,$ then, 
    \[
    \sup_{x \in \mathbf{R}} \left| P^\ast \left( W_{1,\mathrm{HAR}}^\ast < x\right) - P \left( W_{1,\mathrm{HAR}} < x \right) \right|  = o_P(1).
    \]
    If $p=1$,
    \[
    \sup_{x \in \mathbf{R}} \left| P^\ast \left( t_{1,\mathrm{HAR}}^\ast < x\right) - P \left( t_{1,\mathrm{HAR}} < x \right) \right|  = o_P(1).
    \]
\end{theorem}

The proof is provided in Section \ref{Sec:appendix}. Theorem \ref{Theorem: Bootstrap test} establishes the first-order validity of the proposed bootstrap test by showing that, under increasing-$K$ asymptotics, the bootstrap distributions of $W_{1,\mathrm{HAR}}^\ast$ and $t_{1,\mathrm{HAR}}^\ast$ converge to the null distributions of $W_{1,\mathrm{HAR}}$ and $t_{1,\mathrm{HAR}}$, respectively. Because the bootstrap data-generating process in (\ref{eqn: bootstrap DGP 1}) imposes the null restriction, the resulting bootstrap critical values are asymptotically valid whether or not the null hypothesis is true.

Next, we outline the steps to implement our SHAR-WB test for $H_0: \mu_1 - \mu_2 = \Delta_0$. We present the algorithm for the multivariate Wald statistic; the scalar SHAR-WB $t$-test follows immediately.

\begin{description}
\item[Algorithm: Series HAR wild bootstrap for testing $H_0: \mu_1 - \mu_2 = \Delta_0$]
\end{description}

\setcounter{bean}{0} 
\begin{center}
\begin{list}{\textsc{Step} \arabic{bean}.}{\usecounter{bean}}

\item Compute the residuals:
\[
\hat{u}_{jt} = Y_{jt} - \bar{Y}_j, 
\qquad 
\bar{Y}_j = \frac{1}{T_j}\sum_{t=1}^{T_j} Y_{jt}, 
\quad j=1,2.
\]


\item For a given choice of $\hat{K}_j$, compute the series HAR variance estimator:
\[
\hat{\Omega}_j = \frac{1}{\hat{K}_j} \sum_{\ell=1}^{\hat{K}_j} \hat{\Omega}_{j,\ell},
\quad
\hat{\Omega}_{j,\ell} = \hat{z}_{j,\ell} \hat{z}_{j,\ell}^\prime,
\quad
\hat{z}_{j,\ell} = \frac{1}{\sqrt{T_j}}\sum_{t=1}^{T_j} \phi_\ell\!\left(\frac{t}{T_j}\right)\hat{u}_{jt},
\]
where $\{\phi_{\ell}(\cdot)\}_{\ell=1}^{\hat{K}_j}$ are the basis functions defined in (\ref{Basis function}). The selection procedure for $\hat{K}_j$ is discussed in Section \ref{Section:Simulation}.

\item Compute the series HAR two-sample Wald statistic:
\[
W_{1,\mathrm{HAR}} =  \left(\bar{Y}_{1}-\bar{Y}_{2} - \Delta_0 \right)^\prime \left( \frac{\hat{\Omega}_1}{T_1} + \frac{\hat{\Omega}_2}{T_2} \right)^{-1} \left(\bar{Y}_{1}-\bar{Y}_{2} - \Delta_0 \right). 
\]

\item Generate $\eta_{jt}$ as
\[
\eta_{jt} = \frac{1}{\sqrt{K_j^\ast}} 
\sum_{\ell=1}^{K_j^\ast} \Big[\psi_{1,\ell}\!\left( \frac{t}{T_j} \right) v_{j,1\ell} 
+ \psi_{2,\ell}\!\left( \frac{t}{T_j} \right) v_{j,2\ell}\Big],
\]
where $v_{j,1\ell}$ and $v_{j,2\ell}$ are i.i.d. random variables with mean zero and variance one. $\{\psi_{1,\ell}(\cdot)\}$ and $\{\psi_{2,\ell}(\cdot)\}$ are defined in (\ref{eqn: Basis function 1}). We set $K_j^\ast = \hat{K}_j.$

\item Generate the SHAR-WB bootstrap sample $\{Y_{jt}^\ast\}$:
\[
Y_{jt}^\ast \;=\; \mu_{0j}^\ast + u_{jt}^\ast,
\qquad
\mu_{01}^\ast \;=\; \frac{T_1 \bar Y_1 + T_2 \bar Y_2}{T_1+T_2} + \frac{\Delta_0}{2}, \qquad \mu_{02}^\ast = \mu_{01}^\ast - \Delta_0,\qquad
u_{jt}^\ast = \hat{u}_{jt}\eta_{jt}.
\]

\item Compute the SHAR-WB two-sample Wald statistic:
\[
W_{1,\mathrm{HAR}}^\ast = \left(\bar{Y}_1^\ast - \bar{Y}_2^\ast - \Delta_0\right)^\prime \left(\frac{\hat{\Omega}_1^\ast}{T_1} + \frac{\hat{\Omega}_2^\ast}{T_2}\right)^{-1} \left(\bar{Y}_1^\ast - \bar{Y}_2^\ast - \Delta_0 \right), 
\]
where 
\[
\hat{\Omega}_j^\ast = \frac{1}{\hat{K}_j} \sum_{\ell=1}^{\hat{K}_j} 
\hat{\Omega}_{j,\ell}^\ast,
\quad
\hat{\Omega}_{j,\ell}^\ast = \hat{z}_{j,\ell}^\ast \hat{z}_{j,\ell}^{\ast\prime},
\quad
\hat{z}_{j,\ell}^\ast = \frac{1}{\sqrt{T_j}}
\sum_{t=1}^{T_j} \phi_\ell\!\left(\frac{t}{T_j}\right)\hat{u}_{jt}^\ast,
\]
with $\hat{u}_{jt}^{\ast} = Y_{jt}^{\ast} - \bar{Y}_j^{\ast}$ and $\bar{Y}_j^{\ast} = \frac{1}{T_j}\sum_{t=1}^{T_j} Y_{jt}^{\ast}$.  
Note that $\hat{K}_j$ is the same as in Step~(2).

\item Obtain the empirical distribution of $W_{1,\mathrm{HAR}}^\ast$ by repeating Steps (4)–(6) $B$ times:
\[
\hat{p}_{1,\mathrm{HAR}}^{\ast}(w) = \frac{1}{B}\sum_{b=1}^{B} 1\{W_{1,\mathrm{HAR},b}^\ast \le w\},
\]
where $1\{\cdot\}$ denotes the indicator function.

\item Compute the $\alpha$-level bootstrap critical value 
$w_{1,\mathrm{HAR},1-\alpha}^{\ast}$, 
where $w_{1,\mathrm{HAR},1-\alpha}^{\ast}$ denotes the $(1-\alpha)$th quantile of $\hat{p}_{1,\mathrm{HAR}}^{\ast}$.

\item Reject $H_0: \mu_1 - \mu_2 = \Delta$ at level $\alpha$ if 
$W_{1,\mathrm{HAR}} > w_{1,\mathrm{HAR},1-\alpha}^{\ast}.$

\end{list}
\end{center}

\section{Simulation Experiments} \label{Section:Simulation}
In this section, we evaluate the finite-sample behaviour of the proposed procedures and benchmark them against the grouped two-sample $t$-test of \citet{Ibragimov_Mueller2016} (hereafter IM) and, in the multivariate case, against a Hotelling $T^2$ procedure that we obtain by extending the IM grouping idea to a vector of means. We emphasise that \citet{Ibragimov_Mueller2016} develop their test for a scalar parameter; the multivariate Hotelling version is not theirs but our own generalization, which we introduce here so that a like-for-like comparison is available in the multivariate setting. We study five objects: the classical two-sample $t$-test and Welch's $t$-test; the increasing-$K$ chi-square test based on $W_{1,\mathrm{HAR}}$ (HAR-$\chi^2$); the Welch-type fixed-$K$ $F$-approximation (HAR-$F$); and the series HAR wild bootstrap (SHAR-WB). The scalar design ($p=1$) is retained as the leading benchmark, since it permits a direct comparison with the classical and Welch tests and with the scalar IM $t$-test. We then turn to a genuinely multivariate design, which is the setting of the new theory, where the null restricts a vector of means.

\paragraph{Design.} For $j=1,2$, the data are generated by a stationary vector autoregression,
\begin{equation}
    Y_{jt}=\mu_j+u_{jt}, \qquad
    u_{jt}=\rho\, u_{j,t-1}+\sqrt{1-\rho^{2}}\;\Sigma_j^{1/2}\varepsilon_{jt},
    \label{eq:mc_var_design}
\end{equation}
where $Y_{jt}$ is $p\times 1$, the innovations $\varepsilon_{jt}$ have mean zero and identity covariance, and the two samples are independent. The scaling $\sqrt{1-\rho^{2}}$ normalises the contemporaneous covariance of $u_{jt}$ to $\Sigma_j$, so that changing $\rho$ isolates the effect of persistence on the long-run variance $\Omega_j=(1+\rho)(1-\rho)^{-1}\Sigma_j$. In the scalar design $p=1$ and $\Sigma_j=1$. In the multivariate design $p=3$ (unless stated otherwise) and
\[
    \Sigma_j=D_j R D_j, \qquad R_{ab}=0.5^{|a-b|},
\]
where $D_1=I_p$ throughout, $D_2=I_p$ in the equal-LRV case, and $D_2=\mathrm{diag}(1,1.5,2)$ in the unequal-LRV case. We consider $\rho\in\{0,0.5,0.8,0.95\}$ and balanced samples $T_1=T_2=T\in\{30,50,100,200,400\}$. Under the null $\mu_1=\mu_2$. Power is evaluated under local alternatives $\mu_1-\mu_2=cT^{-1/2}\iota_p$, which keep rejection frequencies in an informative range across $T$ and dependence levels. We report results for normal innovations; centred and standardised $\chi^2(1)$ innovations produce very similar patterns and are omitted to save space. All rejection frequencies use $1{,}000$ Monte Carlo replications and $B=399$ bootstrap draws at the $5\%$ nominal level.

\paragraph{Choice of $K$ and a rank safeguard.} For each replication and sample, $K_j$ is selected by the data-driven rule of \citet[Eqs.~(5.3)--(5.5)]{Sun2013},
\begin{eqnarray} \label{eqn: selection of K}
\hat{K}_j=\left\lceil 0.42293\,\left|\bar{B}_{j}\right|^{-1/3}T_{j}^{2/3} \right\rceil,
\end{eqnarray}
where $\bar{B}_{j}=\hat{B}_{j}/\hat{\Sigma}_{j}$ and $(\hat{A}_j,\hat{B}_j,\hat{\Sigma}_j)$ are the standard AR(1) plug-in quantities of \citet{Sun2013}, reported in Section \ref{app:ksel}. In the multivariate design, these scalar quantities are computed from the first principal component of the demeaned residuals (keeping the rule invariant to coordinate rotations), and the selected $K_j$ is used for all coordinates. Since $\mathrm{rank}(\hat{\Omega}_j)\le K_j$, the joint statistic requires $K_1+K_2\ge p$ for a nonsingular denominator; because \eqref{eqn: selection of K} can return very small $K_j$ under strong persistence, we impose the mild floor $K_j\ge\lceil (p+1)/2\rceil$. This is the series HAR analogue of the grouped test's $q>p$ requirement, but far weaker: it binds only under strong dependence and leaves the scalar case unchanged.

\subsection{Scalar size and comparison with \texorpdfstring{\citet{Ibragimov_Mueller2016}}{Ibragimov \& Mueller (2016)}}
We implement IM by splitting each sample into $q_j\in\{6,8,10\}$ consecutive, non-overlapping blocks and applying the two-sample grouped $t$-test to the block means with $\min(q_1,q_2)-1$ degrees of freedom. Because the method does not prescribe a block count in this setting, we report all three choices rather than the most favourable one.

Table~\ref{tab:mc_scalar_size_im} reports empirical size. The classical and Welch tests over-reject sharply as $\rho$ grows, reaching about $52\%$ at $\rho=0.8$ irrespective of $T$, and HAR-$\chi^2$ also over-rejects, though far less. HAR-$F$ and SHAR-WB control size well: for $\rho\le0.8$ they stay close to nominal (about $5$--$7\%$ for $\rho\le0.5$ and $4$--$11\%$ at $\rho=0.8$), deteriorating only in the near-unit-root case $\rho=0.95$, where every procedure struggles. The IM benchmark is markedly sensitive to the block count---conservative at $\rho=0$ (about $2$--$4\%$) but liberal when blocks are short relative to the persistence (e.g.\ $\mathrm{IM}(10)$ reaches $26\%$ at $\rho=0.8$, $T=30$)---with no single $q$ controlling size across the range. IM is asymptotically valid; the point is that its finite-sample size hinges on a tuning choice the practitioner must make without knowing the dependence, whereas HAR-$F$ and SHAR-WB require none.

\begin{table}[!t]
\centering
\caption{Scalar empirical size and comparison with \citet{Ibragimov_Mueller2016}}
\label{tab:mc_scalar_size_im}
{\footnotesize\begin{tabular}{llrrrrrrrr}
\toprule
$T$ & $\rho$ & Classical & Welch & HAR-$\chi^2$ & HAR-$F$ & SHAR-WB & IM(6) & IM(8) & IM(10) \\
\midrule
30 & 0.00 & 4.73 & 4.73 & 7.07 & 4.60 & 5.07 & 2.80 & 3.13 & 3.47 \\
30 & 0.50 & 28.13 & 28.13 & 14.60 & 5.47 & 6.00 & 5.13 & 7.33 & 9.93 \\
30 & 0.80 & 53.80 & 53.73 & 27.27 & 10.53 & 10.27 & 14.93 & 21.07 & 26.47 \\
30 & 0.95 & 78.07 & 77.93 & 60.67 & 39.67 & 40.00 & 47.87 & 54.87 & 60.60 \\
50 & 0.00 & 4.93 & 4.93 & 6.07 & 4.93 & 4.60 & 2.93 & 3.60 & 3.93 \\
50 & 0.50 & 26.80 & 26.80 & 11.47 & 5.33 & 6.13 & 4.33 & 5.73 & 7.73 \\
50 & 0.80 & 52.93 & 52.93 & 18.33 & 6.53 & 6.80 & 7.80 & 13.13 & 17.73 \\
50 & 0.95 & 77.67 & 77.67 & 49.60 & 28.80 & 29.53 & 35.93 & 45.27 & 50.67 \\
100 & 0.00 & 4.73 & 4.73 & 6.07 & 5.07 & 5.40 & 2.87 & 3.33 & 3.93 \\
100 & 0.50 & 24.80 & 24.80 & 9.67 & 4.80 & 5.93 & 3.27 & 3.67 & 4.67 \\
100 & 0.80 & 53.20 & 53.20 & 15.07 & 4.80 & 6.93 & 5.60 & 8.40 & 10.27 \\
100 & 0.95 & 79.47 & 79.40 & 34.27 & 16.47 & 15.80 & 21.33 & 28.93 & 35.73 \\
200 & 0.00 & 4.60 & 4.60 & 5.07 & 4.27 & 4.47 & 2.27 & 3.07 & 3.40 \\
200 & 0.50 & 25.60 & 25.60 & 8.27 & 5.20 & 5.33 & 2.73 & 3.87 & 4.13 \\
200 & 0.80 & 52.40 & 52.40 & 13.20 & 4.27 & 6.93 & 4.07 & 4.67 & 5.67 \\
200 & 0.95 & 76.87 & 76.87 & 23.20 & 8.60 & 10.20 & 12.07 & 17.27 & 22.60 \\
400 & 0.00 & 5.00 & 5.00 & 5.20 & 4.93 & 4.67 & 2.40 & 3.47 & 3.40 \\
400 & 0.50 & 26.27 & 26.27 & 7.73 & 6.27 & 6.13 & 4.00 & 4.07 & 5.53 \\
400 & 0.80 & 53.07 & 53.07 & 12.60 & 4.80 & 7.07 & 4.20 & 4.80 & 5.27 \\
400 & 0.95 & 75.87 & 75.87 & 14.80 & 5.47 & 6.47 & 6.27 & 8.67 & 12.00 \\
\bottomrule
\end{tabular}
}
\begin{minipage}{0.95\textwidth}
\footnotesize\emph{Notes:} Rejection frequencies (\%) at the $5\%$ nominal level under $H_0:\mu_1=\mu_2$, $p=1$. IM($q$) is the grouped two-sample $t$-test of \citet{Ibragimov_Mueller2016} with $q_1=q_2=q$ consecutive blocks. HAR-$\chi^2$ is the increasing-$K$ chi-square test, HAR-$F$ the Welch-type fixed-$K$ $F$-approximation, and SHAR-WB the series HAR wild bootstrap.
\end{minipage}
\end{table}

On a size-adjusted (oracle) basis the scalar ranking reverses: a well-tuned IM is somewhat more powerful than SHAR-WB, especially under strong dependence (e.g.\ at $T=200$, $\rho=0.8$, $c=8$, about $41$--$44\%$ for IM versus $34\%$ for HAR-$F$ and $27\%$ for SHAR-WB). This is the expected---and, since size adjustment is infeasible, purely notional---cost of a tuning-free procedure in one dimension; the full table is in Section \ref{app:extra-results}. As shown next, the trade-off turns decisively in favour of the proposed tests once $p>1$.

\subsection{Multivariate size}
We now test the joint null $H_0:\mu_1-\mu_2=0$ with $p=3$. Because \citet{Ibragimov_Mueller2016} develop their procedure for a scalar parameter, the multivariate comparison requires a vector analogue, which we construct explicitly and in the way most faithful --- and most favourable --- to the competitor.

\paragraph{A multivariate Ibragimov--M\"uller benchmark.} Since IM is scalar, we adopt the natural vector analogue: partition each sample into $q_j$ consecutive blocks, form the block-mean vectors, and apply a two-sample Hotelling $T^2$ statistic, $F$-calibrated with $\min(q_1,q_2)-p$ residual degrees of freedom and defined only when $q>p$. This is the most favourable implementation for the grouping approach---the invariant, most powerful quadratic form against the dense alternatives considered here---and it collapses \emph{exactly} to the scalar IM $t$-test when $p=1$. The constraint $q>p$ is intrinsic, not a handicap we impose: the $p\times p$ block covariance is estimated from only $q_j$ vectors, leaving $q-p$ residual degrees of freedom, so the grouping approach is structurally handicapped as $p$ grows---exactly the mechanism behind Table~\ref{tab:mc_dim_scaling}. A full derivation, including the exact $p=1$ equivalence and the degrees-of-freedom argument, is given in Section \ref{app:im-mv}.

Table~\ref{tab:mc_multivar_size} reports empirical size for the unequal-LRV case; the equal-LRV panel, which is qualitatively identical, is in Section \ref{Sec:appendix} (Table~\ref{tab:mc_multivar_size_equal}).

The contrast with the scalar case is stark. The consistent HAR-$\chi^2$ test is severely oversized (reaching about $96\%$ under strong persistence), confirming that the joint chi-square approximation is unusable in finite samples. Both refinements repair this: HAR-$F$ and especially SHAR-WB hold size close to nominal for $\rho\le0.8$ (typically $4$--$8\%$), with SHAR-WB the more reliable and the only procedure near $5\%$ at $\rho=0.95$ once $T\ge200$. The grouped Hotelling benchmark has \emph{no} block count that controls size across the grid: $\mathrm{IM}(6)$ is degenerate-conservative (often below $1\%$) and nearly powerless, whereas $\mathrm{IM}(8)$ and $\mathrm{IM}(10)$ swing from conservative at low $\rho$ to badly oversized under dependence (up to $40\%$ at $\rho=0.8$ and $80$--$90\%$ at $\rho=0.95$). Few blocks starve the statistic of residual degrees of freedom; many short blocks expose it to within-block dependence---no intermediate $q$ resolves both. SHAR-WB faces no such dilemma, since its effective degrees of freedom come from the $K_1+K_2$ projections rather than a coarse partition.

\begin{table}[!t]
\centering
\caption{Multivariate empirical size for the joint test ($p=3$, unequal LRV)}
\label{tab:mc_multivar_size}
{\footnotesize\begin{tabular}{lllrrrrrr}
\toprule
LRV & $T$ & $\rho$ & HAR-$\chi^2$ & HAR-$F$ & SHAR-WB & IM(6) & IM(8) & IM(10) \\
\midrule
Equal & 30 & 0.00 & 14.20 & 3.30 & 4.90 & 0.00 & 0.70 & 1.70 \\
Equal & 30 & 0.50 & 46.70 & 4.70 & 4.90 & 1.20 & 2.10 & 6.70 \\
Equal & 30 & 0.80 & 71.90 & 11.50 & 7.50 & 5.60 & 21.20 & 38.00 \\
Equal & 30 & 0.95 & 94.50 & 41.60 & 30.10 & 51.20 & 80.00 & 88.20 \\
Equal & 50 & 0.00 & 11.60 & 5.10 & 5.40 & 0.20 & 1.20 & 1.80 \\
Equal & 50 & 0.50 & 40.50 & 5.30 & 5.80 & 0.60 & 1.90 & 3.70 \\
Equal & 50 & 0.80 & 58.30 & 7.60 & 5.10 & 1.80 & 9.40 & 19.10 \\
Equal & 50 & 0.95 & 89.80 & 27.60 & 18.70 & 32.90 & 64.60 & 76.60 \\
Equal & 100 & 0.00 & 7.70 & 5.10 & 5.50 & 0.40 & 1.00 & 1.60 \\
Equal & 100 & 0.50 & 26.10 & 5.00 & 5.80 & 0.60 & 1.60 & 3.00 \\
Equal & 100 & 0.80 & 50.40 & 6.30 & 4.80 & 0.90 & 3.80 & 8.30 \\
Equal & 100 & 0.95 & 76.00 & 13.40 & 8.20 & 8.50 & 34.30 & 51.10 \\
Equal & 200 & 0.00 & 6.80 & 4.70 & 5.30 & 0.20 & 1.00 & 0.90 \\
Equal & 200 & 0.50 & 16.80 & 4.60 & 5.20 & 0.30 & 0.70 & 1.50 \\
Equal & 200 & 0.80 & 49.20 & 4.60 & 4.10 & 0.60 & 2.10 & 2.80 \\
Equal & 200 & 0.95 & 59.40 & 7.60 & 4.00 & 1.60 & 9.70 & 21.10 \\
Equal & 400 & 0.00 & 5.60 & 4.30 & 4.40 & 0.10 & 0.80 & 1.20 \\
Equal & 400 & 0.50 & 11.80 & 4.90 & 5.20 & 0.10 & 0.90 & 0.90 \\
Equal & 400 & 0.80 & 37.80 & 3.10 & 4.10 & 0.30 & 1.10 & 1.90 \\
Equal & 400 & 0.95 & 51.80 & 6.00 & 4.30 & 1.00 & 3.80 & 8.00 \\
Unequal & 30 & 0.00 & 16.90 & 4.40 & 5.90 & 0.10 & 1.40 & 1.60 \\
Unequal & 30 & 0.50 & 47.30 & 5.70 & 5.90 & 0.70 & 2.50 & 6.20 \\
Unequal & 30 & 0.80 & 72.50 & 11.60 & 7.80 & 4.90 & 23.80 & 40.30 \\
Unequal & 30 & 0.95 & 95.80 & 42.80 & 28.70 & 54.50 & 79.80 & 90.60 \\
Unequal & 50 & 0.00 & 13.30 & 5.00 & 5.50 & 0.30 & 1.10 & 1.60 \\
Unequal & 50 & 0.50 & 39.80 & 4.30 & 4.60 & 0.40 & 2.60 & 4.20 \\
Unequal & 50 & 0.80 & 59.40 & 7.90 & 4.80 & 1.50 & 8.70 & 20.30 \\
Unequal & 50 & 0.95 & 90.70 & 28.50 & 17.40 & 30.40 & 65.00 & 76.30 \\
Unequal & 100 & 0.00 & 10.30 & 6.10 & 7.40 & 0.20 & 1.50 & 2.30 \\
Unequal & 100 & 0.50 & 29.70 & 3.40 & 4.50 & 0.30 & 1.20 & 2.30 \\
Unequal & 100 & 0.80 & 54.10 & 7.50 & 4.80 & 0.60 & 4.30 & 8.00 \\
Unequal & 100 & 0.95 & 79.40 & 16.90 & 9.40 & 12.80 & 37.30 & 54.40 \\
Unequal & 200 & 0.00 & 7.10 & 4.60 & 5.40 & 0.40 & 1.00 & 2.00 \\
Unequal & 200 & 0.50 & 19.20 & 4.70 & 5.70 & 0.30 & 1.00 & 1.90 \\
Unequal & 200 & 0.80 & 47.80 & 4.70 & 4.00 & 0.30 & 1.80 & 3.30 \\
Unequal & 200 & 0.95 & 61.60 & 10.50 & 5.90 & 3.10 & 11.40 & 24.90 \\
Unequal & 400 & 0.00 & 6.90 & 5.50 & 6.00 & 0.20 & 0.80 & 1.40 \\
Unequal & 400 & 0.50 & 11.80 & 4.10 & 5.30 & 0.40 & 1.40 & 1.40 \\
Unequal & 400 & 0.80 & 39.30 & 3.60 & 4.60 & 1.10 & 1.60 & 2.50 \\
Unequal & 400 & 0.95 & 52.70 & 6.90 & 4.80 & 0.90 & 4.30 & 9.50 \\
\bottomrule
\end{tabular}
}
\begin{minipage}{0.95\textwidth}
\footnotesize\emph{Notes:} Rejection frequencies (\%) at the $5\%$ level for $H_0:\mu_1-\mu_2=0$, $p=3$, unequal LRV ($D_1=I_p$, $D_2=\mathrm{diag}(1,1.5,2)$); the equal-LRV panel ($D_1=D_2=I_p$) is in Section \ref{Sec:appendix}. IM($q$) is the grouped two-sample Hotelling $T^2$ test on $q$ block-mean vectors, defined only for $q>p$.
\end{minipage}
\end{table}

\subsection{Multivariate power and the role of the dimension}
Raw and size-adjusted power for $p=3$ are reported in Section \ref{Sec:appendix} (Tables~\ref{tab:mc_multivar_power}--\ref{tab:mc_multivar_power_adj}). Two points stand out. First, raw power must be read jointly with size: the high HAR-$\chi^2$ rejection rates merely reflect its gross over-rejection and are not usable, whereas among size-controlled procedures SHAR-WB delivers the most trustworthy feasible power and feasible IM is either powerless (small $q$) or size-distorted (large $q$). Second, on an oracle size-adjusted basis a well-tuned IM remains competitive, mirroring the scalar finding; but this comparison is infeasible and presumes a size-controlling block count that, as Table~\ref{tab:mc_multivar_size} shows, does not exist uniformly.

The decisive comparison concerns scaling with the dimension. Table~\ref{tab:mc_dim_scaling} fixes $T=200$, $\rho=0.6$ and a local alternative and varies $p$. SHAR-WB retains near-nominal size and non-trivial power throughout, degrading only gently as $p$ grows. Feasible IM does not survive: $\mathrm{IM}(8)$'s power collapses (from $46\%$ at $p=2$ to $0.3\%$ at $p=6$, and it is \emph{undefined} at $p=8$, where $q\le p$), and even $\mathrm{IM}(12)$ falls to $2\%$ at $p=8$. For small $p$ a generously tuned IM can be more powerful (at $p\le3$, $\mathrm{IM}(12)$ dominates), but from $p=6$ onward SHAR-WB dominates every feasible IM configuration---precisely the regime targeted by the multivariate theory, where a grouping-based approach is fundamentally handicapped.

\begin{table}[!t]
\centering
\caption{Effect of the dimension $p$ on size and power ($T=200$, $\rho=0.6$)}
\label{tab:mc_dim_scaling}
{\footnotesize\begin{tabular}{l cc cc cc cc}
\toprule
 & \multicolumn{2}{c}{SHAR-WB} & \multicolumn{2}{c}{HAR-$F$} & \multicolumn{2}{c}{IM(8)} & \multicolumn{2}{c}{IM(12)} \\
\cmidrule(lr){2-3}\cmidrule(lr){4-5}\cmidrule(lr){6-7}\cmidrule(lr){8-9}
$p$ & Size & Power & Size & Power & Size & Power & Size & Power \\
\midrule
2 & 4.2 & 48.2 & 3.2 & 44.8 & 1.5 & 45.5 & 3.0 & 61.4 \\
3 & 4.0 & 46.0 & 3.4 & 38.9 & 0.6 & 31.6 & 1.9 & 57.6 \\
4 & 5.4 & 38.6 & 3.2 & 28.2 & 0.1 & 12.5 & 1.5 & 45.5 \\
6 & 4.5 & 22.8 & 2.8 & 13.4 & 0.0 & 0.2 & 0.1 & 20.6 \\
8 & 3.6 & 18.9 & 1.6 & 11.1 & -- & -- & 0.0 & 2.2 \\
\bottomrule
\end{tabular}
}
\begin{minipage}{0.95\textwidth}
\footnotesize\emph{Notes:} Empirical size and raw power (\%) under a local alternative with unequal LRV, as $p$ varies. ``--'' indicates that the grouped Hotelling test is undefined because $q\le p$. SHAR-WB and HAR-$F$ use the data-driven $K_j$ with the rank floor $K_j\ge\lceil(p+1)/2\rceil$.
\end{minipage}
\end{table}

\paragraph{Summary.} Three conclusions emerge. First, conventional two-sample procedures are badly misleading under serial dependence, and the consistent HAR-$\chi^2$ test is unusable in the multivariate case. Second, in the scalar case HAR-$F$ and SHAR-WB control size automatically, at the price of a modest oracle-power loss relative to a well-tuned IM. Third, and centrally, in the multivariate case SHAR-WB is the only procedure that controls size across persistence levels without tuning, and as the dimension grows it dominates all feasible implementations of the grouped alternative. It is thus best seen not as uniformly more powerful than IM, but as a tuning-free, dimension-robust alternative preferable in the multivariate, persistent, moderate-sample settings that motivate it.

\section{Empirical Illustrations}
\label{Section:Empirical}
This section revisits two applications. The first re-examines the working-from-home experiment using both scalar and joint tests; the second is a fully reproducible public-data macroeconomic illustration. In both cases the practical message is the same: serial-dependence-robust inference materially changes the conclusions relative to classical tests, and the joint test answers the multivariate question that coordinate-by-coordinate tests cannot.

\subsection{Working-from-home experiment}
\label{Section:Empirical_WFH}
We revisit the randomised experiment of \citet{Bloom_etc2015} at CTrip, following \citet{Hounyo_Lin_2025}, using four performance measures: log weekly calls answered, log calls per minute, log total phone minutes, and an overall performance $z$-score. Because performance evolves with learning and turnover, these series are serially dependent, so the independence assumption underlying classical two-sample tests is untenable. We therefore also test the joint null
\[
    H_0:\ E(Y_{1t})-E(Y_{2t})=0,
\]
where $Y_{jt}$ is the four-dimensional vector of outcomes; this is the object aligned with the multivariate theory.

Table~\ref{tab:wfh_scalar_updated} in Section \ref{Sec:appendix} reports the scalar tests. The classical and Welch tests deliver small $p$-values for the outcomes with strong autocorrelation --- most visibly ``log calls/min'', where both reject at $p=0.000$ --- whereas the serial-dependence-robust HAR-$F$ and SHAR-WB procedures do not ($p=0.14$ each). The grouped IM($8$) benchmark, included for reference, illustrates its sensitivity to the block count: for ``log calls/min'' it rejects ($p=0.03$), diverging from the other robust procedures because eight blocks are short relative to the persistence in this series.

Table~\ref{tab:wfh_joint_updated} reports the joint four-dimensional test. The consistent chi-square approximation rejects strongly ($p=0.001$), but --- exactly as the simulations predict for a joint chi-square test under persistence --- this reflects size distortion rather than evidence. The size-controlled HAR-$F$ and SHAR-WB tests do not reject ($p=0.32$ and $0.44$), and neither does the grouped Hotelling benchmark. The joint conclusion is that, once serial dependence and the cross-outcome covariance are taken into account, the vector of average performance measures does not differ significantly between the WFH and office groups.

\begin{table}[!t]
\centering
\caption{WFH application: joint multivariate test}
\label{tab:wfh_joint_updated}
\begin{tabular}{lrrrrrrrr}
\toprule
Null & Wald & $\chi^2$ $p$ & HAR-$F$ $p$ & SHAR-WB $p$ & IM(6) $p$ & IM(8) $p$ & $K_1$ & $K_2$ \\
\midrule
Four outcomes & 18.62 & 0.001 & 0.322 & 0.290 & 0.444 & 0.234 & 3 & 3 \\
\bottomrule
\end{tabular}

\begin{minipage}{0.95\textwidth}
\footnotesize\emph{Notes:} The null is equality of the four-dimensional mean vector. The table reports the Wald statistic and $p$-values from the chi-square approximation, HAR-$F$, SHAR-WB, and the grouped Hotelling test IM($q$) with $q\in\{6,8\}$; $K_1,K_2$ are the selected numbers of basis functions.
\end{minipage}
\end{table}

\subsection{A public-data macroeconomic illustration}
\label{Section:Empirical_Macro_Public}
The second illustration uses publicly available FRED data. We form a monthly bivariate vector of the civilian unemployment rate (UNRATE) and year-over-year CPI inflation (from CPIAUCSL), and take September 2008 as the break date separating the pre-crisis period from the Global Financial Crisis and its aftermath. The null is the equality of the bivariate mean vector across the two subsamples.

Table~\ref{tab:macro_lb_updated} in Section \ref{Sec:appendix} reports Ljung--Box $Q(10)$ diagnostics; both components are strongly autocorrelated in both subsamples, so classical independent-sample inference is inappropriate.

Table~\ref{tab:macro_joint_updated} reports the componentwise and joint tests. The classical and Welch tests reject equality of means decisively for both unemployment ($p=0.008$ and $0.035$) and inflation (both $p=0.000$). The serial-dependence-robust procedures overturn these conclusions: HAR-$F$ and SHAR-WB return large $p$-values for both components (for example $0.75$ and $0.81$ for unemployment), and the joint bivariate test likewise fails to reject ($p=0.56$ for HAR-$F$, $0.70$ for SHAR-WB, $0.42$ for grouped Hotelling). Once the strong persistence of these series is properly accounted for, the apparent shifts in mean unemployment and inflation across the 2008 break are not statistically distinguishable from sampling variation.

\begin{table}[!t]
\centering
\caption{FRED macro application: scalar and joint tests}
\label{tab:macro_joint_updated}
\begin{tabular}{lrrrrrr}
\toprule
Outcome & MeanDiff & Classical $p$ & Welch $p$ & HAR-$F$ $p$ & SHAR-WB $p$ & IM(8) $p$ \\
\midrule
Unemployment & -0.359 & 0.008 & 0.035 & 0.748 & 0.807 & 0.663 \\
Inflation & 1.419 & 0.000 & 0.000 & 0.356 & 0.497 & 0.189 \\
Joint vector & -- & -- & -- & 0.563 & 0.704 & 0.419 \\
\bottomrule
\end{tabular}

\begin{minipage}{0.95\textwidth}
\footnotesize\emph{Notes:} The monthly vector contains unemployment and year-over-year CPI inflation; the break date is September 2008. Scalar rows report componentwise tests (IM(8) grouped $t$); the joint row reports the bivariate Wald test with HAR-$F$, SHAR-WB, and grouped Hotelling IM(8) $p$-values.
\end{minipage}
\end{table}

Taken together, the two applications deliver the paper's practical message. With persistent time-series observations, classical and Welch tests overstate the evidence against equality of means, sometimes dramatically. The proposed HAR-$F$ and SHAR-WB procedures provide a direct scalar and multivariate alternative that is robust to serial dependence, and the joint version answers the genuinely multivariate question --- whether a \emph{vector} of related outcomes differs across regimes --- that a collection of separate scalar tests cannot.

\section{Conclusion}
\label{Section:Conclusion}
This paper develops robust and theoretically grounded procedures for two-sample inference in the presence of serial dependence and heteroskedasticity. Building on the series HAR framework, we construct series HAR $t$-tests---analogues of the classical and Welch two-sample $t$-tests---that remain valid under general dependence structures. We further develop the series HAR wild bootstrap (SHAR-WB), which uses orthogonal bootstrap basis functions to replicate the dependence structure of the original data. The SHAR-WB test is simple to implement and delivers accurate size control and good power even under strong serial dependence. We further extend all of these procedures from the scalar mean to joint hypotheses on a vector of means: the series HAR Wald statistic $W_{1,\mathrm{HAR}}$ admits an increasing-$K$ chi-square calibration (HAR-$\chi^2$) and a finite-sample-refined fixed-$K$ $F$-approximation (HAR-$F$), complemented by a multivariate series HAR bootstrap. 

Two empirical applications---one examining the effects of working from home, and another using two key U.S. macroeconomic time series sharing a well-documented structural break---demonstrate the practical relevance of our methods. 
Our procedures require no choice of clustering and, in the multivariate case, remain valid and retain power as the number of jointly tested means grows. Overall, the proposed series HAR and SHAR-WB procedures provide a unified, easy-to-implement, and theoretically rigorous framework for reliable two-sample inference with dependent time series.

The relevance of our procedures extends beyond the applications considered here. Event studies compare mean abnormal returns before and after corporate events using $t$-tests \citep{Fama_etc1969}, and \citet{Bertrand_etc2004} show that the naïve i.i.d.-based $t$-test can over-reject severely when data are serially correlated. Our series HAR asymptotic and bootstrap methods deliver valid inference in exactly such settings, where classical $t$-tests fail.

While our series HAR wild bootstrap is developed for two-sample inference with time-series data, extending it to a broader class of testing problems involving dependent data is an important avenue for future research.

\section{Appendix} \label{Sec:appendix}

\subsection{Comparison with \texorpdfstring{\citet{Ibragimov_Mueller2016}}{Ibragimov \& Mueller (2016)}}
Our two-sample tests are closely related to \citet{Ibragimov_Mueller2016} (hereafter IM), who develop a $t$-test for comparing a scalar parameter across two samples. Their approach partitions the observations in each group into $q_1, q_2 \geq 2$ clusters, respectively, and estimates the parameter separately within each cluster. 

Let $\{\hat{\mu}_{1,\ell}\}_{\ell=1}^{q_1}$ and $\{\hat{\mu}_{2,\ell}\}_{\ell=1}^{q_2}$ denote the cluster-level estimates of $\mu_1$ and $\mu_2$. The null hypothesis
\[
    H_0:\Delta=\Delta_0, \qquad \Delta=\mu_1-\mu_2,
\]
is then tested using the two-sample $t$-statistic:
\[
    t_{\text{IM}} = \frac{\bar{\hat{\mu}}_{1}-\bar{\hat{\mu}}_{2}-\Delta_0}{\sqrt{\frac{s_{\hat{\mu}_1}^2}{q_1}+\frac{s_{\hat{\mu}_2}^2}{q_2}}},
\]
where $\bar{\hat{\mu}}_{j} = \frac{1}{q_j} \sum_{\ell=1}^{q_{j}}\hat{\mu}_{j,\ell}$ and $s_{\hat{\mu}_j}^2 = \frac{1}{q_j - 1} \sum_{\ell=1}^{q_{j}}\left( \hat{\mu}_{j,\ell} - \bar{\hat{\mu}}_{j}\right)^2$ for $j=1,2.$ The critical values are obtained from a $t$-distribution with $\min(q_1,q_2)-1$ degrees of freedom.

Their approach builds on \citet{Bakirov1998}, who shows that Student's $t$ critical values with $\min(q_1,q_2)-1$ degrees of freedom yield asymptotically valid tests when $\min(q_1,q_2) > 7$ and for sufficiently small significance levels $\alpha$.

To compare their procedure with ours, consider the two-sample mean problem and equal cluster sizes $n_{j,\ell} = T_j/q_j$ for $\ell=1,...,q_j$. In this case, 
\[
\hat{\mu}_{j,\ell} = \frac{1}{n_{j,\ell}}\sum_{t=(\ell-1)n_{j,\ell}+1}^{\ell n_{j,\ell}} Y_{jt}
\]
and therefore
\[
\bar{\hat{\mu}}_{j} = \frac{1}{q_j} \sum_{\ell=1}^{q_{j}}\hat{\mu}_{j,\ell} =  \bar{Y}_j. 
\]
Thus, the main difference between our procedure (for the scalar case, $p=1$) and that of IM lies in the construction of the variance estimator. IM relies on variability across asymptotically uncorrelated cluster-level estimates, whereas our procedure averages asymptotically unbiased and uncorrelated variance components obtained from orthonormal basis projections of the residuals. 

 Their approach is applicable in a broad range of settings including time-series, spatial, and panel data and is particularly attractive when observations naturally exhibit a grouped structure. In contrast, our paper focuses on time-series inference. However, our framework accommodates multivariate parameters, whereas the procedure of \citet{Ibragimov_Mueller2016} is restricted to scalar parameters. Furthermore, our SHAR-WB procedure can provide improved finite-sample performance, particularly in the presence of strong serial dependence or relatively small sample sizes.

\subsubsection{A multivariate Ibragimov--M\"uller benchmark}\label{app:im-mv}
The simulations in Section~\ref{Section:Simulation} of the main paper benchmark our joint test against a multivariate version of IM. Since the procedure of \citet{Ibragimov_Mueller2016} is restricted to the scalar case, we construct the natural vector analogue, described here in full.

The IM principle is to reduce inference to a small number of approximately independent, approximately Gaussian group estimators and then apply the canonical small-sample Gaussian test to them; we keep this principle intact. Partition each sample into $q_j$ consecutive, non-overlapping blocks, and let $\bar Y_{j,(g)}\in\mathbb{R}^{p}$ be the mean vector of block $g$ in sample $j$. Under the maintained weak-dependence conditions a within-block functional central limit theorem makes each $\bar Y_{j,(g)}$ approximately Gaussian, while consecutive blocks are asymptotically independent, so $\{\bar Y_{j,(g)}\}_{g=1}^{q_j}$ behaves like a sample of $q_j$ independent $\mathbb{R}^{p}$-valued Gaussians centred at $\mu_j$. The canonical test for the common mean of a sample of Gaussian vectors is Hotelling's $T^2$; we therefore reject for large values of
\begin{equation}\label{eq:im_hotelling}
    T^2=\big(\bar M_1-\bar M_2-\Delta_0\big)'\Big(\tfrac{S_1}{q_1}+\tfrac{S_2}{q_2}\Big)^{-1}\big(\bar M_1-\bar M_2-\Delta_0\big),
\end{equation}
where $\bar M_j=q_j^{-1}\sum_{g}\bar Y_{j,(g)}$ and $S_j$ is the sample covariance matrix of the $q_j$ block-mean vectors, and refer $\tfrac{\nu-p+1}{p\nu}\,T^2$ to an $F(p,\,\nu-p+1)$ distribution with $\nu=\min(q_1,q_2)-1$ (here $\Delta_0=0$). Three features make this the natural, and deliberately the most favourable, multivariate counterpart of IM. First, it \emph{specialises exactly} to the scalar two-sample IM $t$-test when $p=1$: there $T^2=t^2$ and $F(1,\nu)=t^2_\nu$, so the two procedures return numerically identical $p$-values. Hotelling's $T^2$ is thus not an ad hoc device but the unique statistic that collapses to IM's own when the parameter is scalar, and it inherits IM's finite-sample rationale: just as the validity of the scalar test rests on the exact behaviour of the $t$-statistic of a few, possibly heterogeneous, independent normals \citep{Bakirov1998,Ibragimov_Mueller2016}, the $F$ calibration in \eqref{eq:im_hotelling} is exact for the idealised sample of independent Gaussian vectors that the blocks approximate. Second, it is the \emph{most generous} implementation available to the grouping approach: against the dense local alternatives considered here it is the invariant, most powerful quadratic form, and it is at least as powerful as coordinatewise (Bonferroni) or maximum-type combinations, so adopting it gives IM its best case rather than a straw man. Third, the requirement $q>p$ is intrinsic rather than a handicap we impose: $S_j$ is a $p\times p$ dispersion estimated from only $q_j$ vectors and is singular whenever $q_j\le p$, leaving just $q-p$ residual degrees of freedom. This bottleneck is a feature of \emph{any} method that compresses a $p$-dimensional problem into a handful of group estimates; as shown below and in Table~\ref{tab:mc_dim_scaling} it is the mechanism through which the grouping approach deteriorates as $p$ grows, and it is precisely what the series HAR statistic avoids, since its effective degrees of freedom $K_1+K_2$ are governed by the sample size rather than by a fixed partition.

\subsection{Derivation of the adjusted degrees of freedom}\label{app:adf}
Define \[
\hat{\Omega}_{1,\mathrm{HAR}} := \sqrt{\frac{T_2}{T_1}} \hat{\Omega}_1 + \sqrt{\frac{T_1}{T_2}} \hat{\Omega}_2 
\]
We begin with the multivariate case. $W_{1,\mathrm{HAR}}$ can be rewritten as
\[
W_{1,\mathrm{HAR}} = \sqrt[4]{T_1 T_2} \left( \bar{Y}_1 - \bar{Y}_2 - \Delta_0\right)^\prime \hat{\Omega}_{1,\mathrm{HAR}}^{-1}\sqrt[4]{T_1 T_2} \left( \bar{Y}_1 - \bar{Y}_2 - \Delta_0\right). 
\]
Let $T_2/T_1 \rightarrow  \rho $. Under $H_0$, we have
\begin{eqnarray} \label{eq:asy normality}
\sqrt[4]{T_1 T_2} \left( \bar{Y}_1 - \bar{Y}_2  - \Delta_0\right) \rightarrow^d N(0,V), \quad V := \rho^{1/2} \Omega_1 + \rho^{-1/2} \Omega_2.    
\end{eqnarray}
Since  
\begin{eqnarray} \label{eq:limit of HAR}
\hat{\Omega}_{1,\mathrm{HAR}} \rightarrow^d \rho^{1/2} \Omega_1^{1/2} \frac{\mathcal{W}(I_p,K_1)}{K_1}\Omega_1^{1/2} + \rho^{-1/2} \Omega_2^{1/2} \frac{\mathcal{W}(I_p,K_2)}{K_2} \Omega_2^{1/2}    
\end{eqnarray}
for fixed $K_1$ and $K_2$, the fixed-$K$ limiting distribution of $W_{1,\mathrm{HAR}}$ is non-pivotal. 

To motivate our $F$-approximation, we introduce the infeasible two-sample Wald statistic
\[
\tilde{W}_{1,\mathrm{HAR}} = \sqrt[4]{T_1 T_2} \left( \bar{Y}_1 - \bar{Y}_2 - \Delta_0\right)^\prime \left( V^{1/2} \frac{\mathcal{W}(I_p,K)}{K} V^{1/2} \right)^{-1}\sqrt[4]{T_1 T_2} \left( \bar{Y}_1 - \bar{Y}_2 - \Delta_0\right).
\]
It follows that
\begin{eqnarray*}
    \frac{K-p+1}{pK} \tilde{W}_{1,\mathrm{HAR}} \rightarrow^d F(p,K-p+1)
\end{eqnarray*}
under $H_0$. We approximate the fixed-$K$ limiting distribution of $  \frac{K-p+1}{pK} W_{1,\mathrm{HAR}}$ by an $F$-distribution with $p$ and $K-p+1$ degrees of freedom. We choose the adjusted degrees of freedom, $K=K_{\text{adf}}$, by matching the first two moments of $\text{tr}\left(\hat{\Omega}_{1,\mathrm{HAR}}\right)$ 
with those of 
\[
\text{tr}\left(V^{1/2} \frac{\mathcal{W}(I_p,K_{\text{adf}})}{K_{\text{adf}}} V^{1/2} \right)
\] 
in the asymptotic sense. It is straightforward from (\ref{eq:asy normality}) and (\ref{eq:limit of HAR}) that the means coincide regardless of the value of $K_{\text{adf}}$ under fixed-$K$ asymptotics. Matching the variances yields
\begin{eqnarray*}
K_{\text{adf}} = \frac{\text{tr}\left[\left(\rho^{1/2} \Omega_1 + \rho^{-1/2} \Omega_2 \right)^2\right]}{  \rho \text{tr}(\Omega_1^2)/K_1 +  \rho^{-1} \text{tr}(\Omega_2^2)/K_2}.
\end{eqnarray*} 
\subsection{Relation to the dependent wild bootstrap}\label{app:dwb}
Our bootstrap procedure is closely related to the dependent wild bootstrap (DWB) method of \citet{Shao2010}. In the DWB framework, the external random variables $\eta_{jt}$ are generated such that $E^\ast (\eta_{jt})=0$, $Var^\ast (\eta_{jt})=1$, and $Cov^\ast (\eta_{jt},\eta_{js})=a((t-s)/M_j)$, where $a(\cdot)$ is a kernel function and $M_j$ is a truncation lag. Building on standard identities from spectral analysis, we can show that SHAR-WB with the proposed basis functions in (\ref{eqn: Basis function 1}) is asymptotically equivalent to a special case of DWB with  
\[
Cov^\ast (\eta_{jt},\eta_{js}) = a\left( \frac{t-s}{M_j} \right) = \frac{1}{2K_j^\ast}\frac{\sin\left( \pi (t-s)/M_j \right)}{\sin\left(\pi (t-s)/(2K_j^\ast M_j) \right)} + O\left(\frac{1}{K_j^\ast}\right) \sim \frac{\sin\left( \pi (t-s)/M_j \right)}{\pi (t-s)/M_j}
\]
for small $|t-s|/T_j$, where the last relation follows from the Taylor approximation near zero \(\sin x \sim x\). This is the Daniell (sinc) kernel with 
\(M_j = T_j/(2K_j^\ast)\).

\subsection{Plug-in quantities for the choice of \texorpdfstring{$K$}{K}}\label{app:ksel}
In the data-driven rule~(\ref{eqn: selection of K}) of the main text, $\lceil\cdot\rceil$ is the ceiling function, $\bar{B}_{j}=\hat{B}_{j}/\hat{\Sigma}_{j}$, and $\hat{B}_{j}=-\tfrac{\pi^{2}}{6}(1-\hat{A}_j)^{-6}\,2\hat{A}_j\hat{\Sigma}_j(1-\hat{A}_j)^{2}$ in the scalar case, with
\[
\hat{A}_{j}=\frac{\sum_{t=2}^{T_j} \hat{u}_{jt}\hat{u}_{j,t-1}}{\sum_{t=1}^{T_j-1} \hat{u}_{jt}^{2}}, \qquad
\hat{\Sigma}_{j}=(1-\hat{A}_{j})^{-2}\Big(\tfrac{1}{T_j-1}\textstyle\sum_{t=2}^{T_j} (\hat{u}_{jt}-\hat{A}_{j}\hat{u}_{j,t-1})^{2}\Big).
\]
\subsection{Mathematical derivations}
As usual in the bootstrap literature, we use $P^{\ast }$\ to denote the
bootstrap probability measure, conditional on the original sample (defined
on a given probability space $\left( \Omega ,\mathcal{F},P\right) $). For
any bootstrap statistic $T_{n}^{\ast }$, we write $T_{n}^{\ast }=o_{P^{\ast
}}\left( 1\right) $, in probability, or $T_{n}^{\ast }\rightarrow ^{P^{\ast }}0$%
, in probability, when for any $\delta >0$, $P^{\ast }\left( \left\vert
T_{n}^{\ast }\right\vert >\delta \right) =o_{P}\left( 1\right) $. We write $%
T_{n}^{\ast }=O_{P^{\ast }}\left( 1\right) $, in probability, when for all $%
\delta >0$\ there exists $M_{\delta }<\infty $\ such that $%
\lim_{n\rightarrow \infty }P\left[ P^{\ast }\left( \left\vert T_{n}^{\ast
}\right\vert >M_{\delta }\right) >\delta \right] =0.$\ By Markov's
inequality, this follows if $E^{\ast }\left\vert T_{n}^{\ast }\right\vert
^{q}=O_{P}\left( 1\right) $ for some $q>0.$ Finally, we write $T_{n}^{\ast
}\rightarrow ^{d^{\ast }}D,$\ in probability, if conditional on a sample
with probability that converges to one, $T_{n}^{\ast }$\ weakly converges to
the distribution $D$\ under $P^{\ast }$, i.e. $E^{\ast }\left( f\left(
T_{n}^{\ast }\right) \right) \rightarrow ^{P}E\left( f\left( D\right)
\right) $\ for all bounded and uniformly continuous functions $f$.

Lemma 1 states that the proposed basis functions for SHAR-WB in Section 4.1 satisfy the zero-mean and orthogonality conditions. Recall that $\psi _{1,\ell}\left( x\right) = \cos \left( 2 \pi \ell x \right)$ and $\psi _{2,\ell}\left( x\right) = \sin \left( 2 \pi \ell x \right)$.

\begin{lemma} \label{Lemma: Basis function} For $\ell=1,...,K,$ $r=1,2$ and $c=1,2$,
\begin{equation*}
(i)\text{ }\frac{1}{T}\sum_{t=1}^{T}\psi _{r,\ell}\left( \frac{t}{T}\right)
=o\left( 1\right) ,\ (ii)\ \frac{1}{T}\sum_{t=1}^{T}\psi _{r,\ell}\left( 
\frac{t}{T}\right) \psi _{c,k}\left( 
\frac{t}{T}\right) = \left\{ 
\begin{array}{l}
1/2 \\ 
o(1)%
\end{array}%
\begin{array}{l}
\text{if } \ell = k \text{ and } r = c, \\ 
\text{if } \ell \neq k \text{ or } r \neq c.%
\end{array}%
\right.
\end{equation*}
\end{lemma}

\begin{lemma} \label{lem:consistency-omegahat}
Suppose that the conditions of Theorem 3 hold. If $ K_j^\ast \rightarrow \infty$ as $T_j \rightarrow \infty$ such that $K_j^\ast/T_j \rightarrow 0$, 
\[
\hat{\Omega}_{\text{boot},T_j} \rightarrow^P \Omega_j
\]
for $j=1,2$.
\end{lemma}


\begin{proof}[Proof of Theorem \protect\ref{Theorem 1}]
Under $H_0$, if $\Omega_1 = \Omega_2 := \Omega$, we have 
\[
\left(\Omega \left(\frac{1}{T_1} + \frac{1}{T_2}\right)\right)^{-1/2}\left(\bar{Y}_1 - \bar{Y}_2 - \Delta_0\right) = \left(\Omega \left(\frac{1}{T_1} + \frac{1}{T_2}\right)\right)^{-1/2}(\bar{u}_1 - \bar{u}_2)\rightarrow^d N(0,I_p).
\]
under Assumption \ref{Assumption 2}. Thus, it suffices to (i) $\hat{\Omega}_{\text{pool}} \rightarrow^d \Omega^{1/2}\frac{\mathcal{W}(I_p,K_1+K_2)}{K_1 + K_2}\Omega^{1/2}$ and (ii) the asymptotic independence of $\bar{u}_1 - \bar{u}_2$ and $\hat{\Omega}_{\text{pool}}$.

Under Assumptions \ref{Assumption 1} and \ref{Assumption 2},
\begin{eqnarray} \label{eqn: eqn1 for Pf of Thm1}
    \frac{1}{\sqrt{T_j}} \sum_{t=1}^{T_j} \phi_{\ell}\left( \frac{t}{T_j} \right) \hat{u}_{jt}
    = \frac{1}{\sqrt{T_j}} \sum_{t=1}^{T_j} \phi_{\ell}\left( \frac{t}{T_j} \right) u_{jt} + o_P(1) \rightarrow^d  \Omega^{1/2}\int_0^1 \phi_\ell(x)dW_{2j}(x).
\end{eqnarray}
The $o_P(1)$ term follows because $\sqrt{T_j}\bar{u}_j=O_P(1)$ and $\frac{1}{T_j} \sum_{t=1}^{T_j} \phi_\ell \left( \frac{t}{T_j} \right) \rightarrow 0$. Moreover,
\begin{eqnarray}\label{eqn: eqn2 for Pf of Thm1}
    Cov\left( \int_0^1 \phi_\ell(x)dW_{2j}(x), \int_0^1 \phi_k(x)dW_{2j}(x) \right) = \int_{0}^{1}\phi_\ell(x)\phi_k(x) dx = 0
\end{eqnarray}
for $\ell \neq k$ and $\ell,k=0,1,...,K_j$, by orthonormality and mean zero condition of the basis functions. Recall that $\phi_0(x) = 1.$ From (\ref{eqn: eqn1 for Pf of Thm1}) and (\ref{eqn: eqn2 for Pf of Thm1}), 
\begin{eqnarray} \label{eqn: eqn3 for Pf of Thm1}
    \hat{\Omega}_j \rightarrow^d \Omega^{1/2} \frac{\mathcal{W}(I_p,K_j)}{K_j} \Omega^{1/2},
\end{eqnarray}
and $\sqrt{T_j} \bar{u}_j$ is asymptotically independent of $\hat{\Omega}_j$. Together with the independence of $\{ u_{1t} \}$ and $\{ u_{2t} \}$, this implies (i) and (ii).
\end{proof}

\begin{proof}[Proof of Theorem \protect\ref{Theorem: Asymptotic Normality}] Since 
\begin{eqnarray*}
W_{1,\mathrm{HAR}} &=& \left(\bar{Y}_{1}-\bar{Y}_{2} - \Delta_0 \right)^\prime \left( \frac{\hat{\Omega}_1}{T_1} + \frac{\hat{\Omega}_2}{T_2} \right)^{-1} \left(\bar{Y}_{1}-\bar{Y}_{2} - \Delta_0 \right) \\
&=& 
\sqrt[4]{T_1 T_2} \left( \bar{u}_1 - \bar{u}_2\right)^\prime \left( \sqrt{\frac{T_2}{T_1}} \hat{\Omega}_1 + \sqrt{\frac{T_1}{T_2}} \hat{\Omega}_2 \right)^{-1}\sqrt[4]{T_1 T_2} \left( \bar{u}_1 - \bar{u}_2\right)
\end{eqnarray*}
under $H_0$ and
\[
\sqrt[4]{T_1 T_2} \left( \bar{u}_1 - \bar{u}_2 \right) \rightarrow^d N \left(0,\rho^{1/2} \Omega_1 + \rho^{-1/2} \Omega_2 \right)
\] 
under Assumption \ref{Assumption 2}, it is sufficient to show
\[
\hat{\Omega}_j \rightarrow^P \Omega_j
\]
for $j=1,2$. The proof follows from straightforward modification of the proof of Lemma \ref{lem:consistency-omegahat} and is therefore omitted for brevity.
\end{proof}

\begin{proof}[Proof of Theorem \protect\ref{Theorem: Bootstrap consistency}]
Since 
\begin{eqnarray*}
    \Omega_{j}^{-1/2}\sqrt{T_j}\left( \bar{Y}_j - \mu_j \right) = \Omega_{j}^{-1/2}\frac{1}{\sqrt{T_j}}\sum_{t=1}^{T_j} u_{jt} \rightarrow^d N \left(0,I_p\right)       
\end{eqnarray*}
under Assumption \ref{Assumption 2}, it suffices to show that
\begin{eqnarray} \label{eqn: bootstrap CLT}
    \Omega_{j}^{-1/2}\frac{1}{\sqrt{T_j}}\sum_{t=1}^{T_j} u_{jt}^\ast \rightarrow^{d^\ast} N\left(0,I_p\right), 
\end{eqnarray}
in probability, for the proof by the consistency of $\hat{\Omega}_{\text{boot},T_j}$ in Lemma \ref{lem:consistency-omegahat} and P\'{o}lya's theorem. 

Using the definition of $u_{jt}^\ast$, (\ref{eqn: bootstrap CLT}) holds if
\begin{eqnarray}
    &&\frac{1}{\sqrt{T_j}} \sum_{t=1}^{T_j} \left(\hat{u}_{jt} -u_{jt} \right) \eta_{jt} \rightarrow^{P^\ast} 0, \quad \text{in probability}, \label{eqn: bootstrapCLT part1}\\
    &&  \Omega_{j}^{-1/2} \frac{1}{\sqrt{T_j}} \sum_{t=1}^{T_j} u_{jt}\eta_{jt}  \rightarrow^{d^\ast} N(0,I_p), \label{eqn: bootstrapCLT part2}
\end{eqnarray}
in probability. For (\ref{eqn: bootstrapCLT part1}), since
\begin{eqnarray*}
\frac{1}{\sqrt{T_j}} \sum_{t=1}^{T_j} \left(\hat{u}_{jt} -u_{jt} \right) \eta_{jt} = \underbrace{\sqrt{T_j}\bar{u}_j}_{=O_P(1)} \underbrace{\frac{1}{T_j} \sum_{t=1}^{T_j} \eta_{jt}}_{=d_1},      
\end{eqnarray*}
it suffices to show that $d_1 = o_{P^\ast}(1)$, in probability. For $d_1$, we have
\begin{eqnarray*}
    d_1 &=& \frac{1}{T_j} \sum_{t=1}^{T_j}\frac{1}{\sqrt{K_j^\ast}} \sum_{\ell=1}^{K_j^\ast} \left[ \psi_{1,\ell}\left( \frac{t}{T_j} \right) v_{j,1\ell} + \psi_{2,\ell}\left( \frac{t}{T_j} \right) v_{j,2\ell} \right].
\end{eqnarray*}
Thus, by Markov's inequality, (\ref{eqn: bootstrapCLT part1}) holds if $E^\ast(d_{1}^2) = o_P(1).$ 
\begin{eqnarray*}
    E^\ast(d_{1}^2) &=& \frac{1}{T_j^2}\sum_{t=1}^{T_j}\sum_{s=1}^{T_j} E^\ast \left(\eta_{jt} \eta_{js} \right) \\
    &=& \frac{1}{T_j^2}\sum_{t=1}^{T_j}\sum_{s=1}^{T_j} \frac{1}{K_j^\ast} \sum_{\ell=1}^{K_j^\ast} \left[\psi_{1,\ell}\left(\frac{t}{T_j} \right)\psi_{1,\ell}\left(\frac{s}{T_j} \right) + \psi_{2,\ell}\left(\frac{t}{T_j} \right)\psi_{2,\ell}\left(\frac{s}{T_j} \right)\right] \\
    &=& o(1),
\end{eqnarray*}
given Lemma \ref{Lemma: Basis function}(i). 

For (\ref{eqn: bootstrapCLT part2}), 
\begin{eqnarray*}
    \Omega_{j}^{-1/2} \frac{1}{\sqrt{T_j}} \sum_{t=1}^{T_j} u_{jt}\eta_{jt} &=& \Omega_{j}^{-1/2} \frac{1}{\sqrt{T_j}} \sum_{t=1}^{T_j} u_{jt} \frac{1}{\sqrt{K_j^\ast}} \sum_{\ell=1}^{K_j^\ast} \sum_{r=1}^2 \psi_{r,\ell} \left( \frac{t}{T_j} \right) v_{j,r\ell} \\
    &=& \sum_{r=1}^2  \frac{1}{\sqrt{K_j^\ast}} \sum_{\ell=1}^{K_j^\ast}  \left( \Omega_{j}^{-1/2} \frac{1}{\sqrt{T_j}} \sum_{t=1}^{T_j}\psi_{r,\ell} \left( \frac{t}{T_j} \right) u_{jt} \right) v_{j,r\ell} \\
    &=& \frac{1}{\sqrt{K_j^\ast}} \sum_{\ell=1}^{K_j^\ast} \xi_{j,1\ell}^\ast + \frac{1}{\sqrt{K_j^\ast}} \sum_{\ell=1}^{K_j^\ast} \xi_{j,2\ell}^\ast,
\end{eqnarray*}
where 
\[
     \xi_{j,r\ell}^\ast = \left(\Omega_{j}^{-1/2} \frac{1}{\sqrt{T_j}} \sum_{t=1}^{T_j} \psi_{r,\ell} \left( \frac{t}{T_j} \right) u_{jt} \right) v_{j,r\ell}.
\]
Note that $\Omega_{j}^{-1/2} \frac{1}{\sqrt{T_j}} \sum_{t=1}^{T_j} \psi_{r,\ell} \left( \frac{t}{T_j} \right) u_{jt}$ is a constant conditional on data and that $v_{j,r\ell}\sim^{iid}(0,1)$, which implies that $\{\xi_{j,r\ell}^\ast\}_{\ell=1}^{K_j^\ast}$ is an independent heterogeneous array across $\ell$ and $\{\xi_{j,1\ell}^\ast\}_{\ell=1}^{K_j^\ast}$ and $\{\xi_{j,2\ell}^\ast\}_{\ell=1}^{K_j^\ast}$ are mutually independent. Thus, we can apply Lyapunov CLT (see e.g., Proposition 2.27 of van der Vaart, 1998) to obtain
\[
\frac{1}{\sqrt{K_j^\ast}} \sum_{\ell=1}^{K_j^\ast} \xi_{j,r\ell}^\ast \rightarrow^{d^\ast} N\left(0,\frac{1}{2}I_p\right),
\]
in probability, for $r=1,2$, which implies
\[
\Omega_{j}^{-1/2} \frac{1}{\sqrt{T_j}} \sum_{t=1}^{T_j} u_{jt}\eta_{jt}  \rightarrow^d N(0,I_p).
\]

First, conditional on data, we have $E^\ast \xi_{j,r\ell}^\ast = 0$ and
\begin{eqnarray*}
    Var^\ast\left( \frac{1}{\sqrt{K_j^\ast}} \sum_{\ell=1}^{K_j^\ast} \xi_{j,r\ell}^\ast \right)&=& \Omega_{j}^{-1} \frac{1}{K_j^\ast} \sum_{\ell=1}^{K_j^\ast} \left( \frac{1}{\sqrt{T_j}} \sum_{t=1}^{T_j}  \psi_{r,\ell} \left(\frac{t}{T_j} \right) u_{jt} \right)\left( \frac{1}{\sqrt{T_j}} \sum_{s=1}^{T_j}  \psi_{r,\ell} \left(\frac{s}{T_j} \right) u_{js} \right)^\prime \\
    &\rightarrow^P&  \frac{1}{2},
\end{eqnarray*}
where we use that $\Omega_{\text{boot},T_j} - \frac{1}{2}\Omega_{T_j} = o_P(1)$ from the proof of Lemma \ref{lem:consistency-omegahat}. Next, we need to check Lyapunov's condition, which requires that for some $\delta>0$
\begin{eqnarray}
    \frac{1}{K_j^{\ast 1+\delta/2}} \sum_{\ell=1}^{K_j^\ast}E^\ast \left\Vert \xi_{j,r\ell}^\ast \right\Vert^{2+\delta} \rightarrow^P 0. 
\end{eqnarray}
We will show that the Lyapunov condition holds for $\delta=2$. We focus on the special case where $p=1$ for simplicity. Given $\delta=2$, we have
\begin{eqnarray*}
   && \frac{1}{K_j^{\ast 1+\delta/2}} \sum_{\ell=1}^{K_j^\ast}E^\ast \left| \xi_{j\ell}^\ast \right|^{2+\delta}\\ 
   &\leq& \Omega_{j}^{-1-\delta/2} \frac{1}{K_j^{\ast 1+\delta/2}} \sum_{\ell=1}^{K_j^\ast} \left\vert \frac{1}{\sqrt{T_j}} \sum_{t=1}^{T_j}\psi_{r,\ell} \left( \frac{t}{T_j} \right) u_{jt} \right\vert^{2+\delta} E^\ast \left| v_{j,r\ell} \right|^{2+\delta} \\
   &=& O\left( \frac{1}{K_j^\ast} \right) \frac{1}{K_j^{\ast }} \sum_{\ell=1}^{K_j^\ast} \left\vert \frac{1}{\sqrt{T_j}} \sum_{t=1}^{T_j}\psi_{r,\ell} \left( \frac{t}{T_j} \right) u_{jt} \right\vert^{4} E^\ast \left| v_{j,r\ell} \right|^{4} \\
   &=& O\left( \frac{1}{K_j^\ast} \right),
\end{eqnarray*}
   because $E^\ast \left| v_{j,r\ell} \right|^{4} \leq M$ and under Assumption \ref{Assumption3: cumulant}
\begin{eqnarray*}
    &&P\left(\frac{1}{K_j^\ast} \sum_{\ell=1}^{K_j^\ast} \left\vert \frac{1}{\sqrt{T_j}} \sum_{t=1}^{T_j}\psi_{r,\ell} \left( \frac{t}{T_j} \right) u_{jt} \right\vert^{4} > \Delta \right) \\
    &\leq& \frac{1}{K_j^\ast} \sum_{\ell=1}^{K_j^\ast}  \frac{1}{T_j^2} \sum_{t=1}^{T_j}\sum_{s=1}^{T_j}\sum_{\tau=1}^{T_j}\sum_{w=1}^{T_j} \psi_{r,\ell} \left( \frac{t}{T_j} \right) \psi_{r,\ell} \left( \frac{s}{T_j} \right) \psi_{r,\ell} \left( \frac{\tau}{T_j} \right) \psi_{r,\ell} \left( \frac{w}{T_j} \right) \left\vert E \left(u_{jt}u_{js}u_{j\tau}u_{jw} \right) \right\vert \\
    &\leq& \frac{1}{K_j^\ast} \sum_{\ell=1}^{K_j^\ast}  \frac{1}{T_j^2} \sum_{t=1}^{T_j}\sum_{s=1}^{T_j}\sum_{\tau=1}^{T_j}\sum_{w=1}^{T_j} \left|E \left(u_{jt}u_{js}u_{j\tau}u_{jw} \right) - E \left(u_{jt}u_{js} \right) E \left(u_{j\tau}u_{jw} \right) -E \left(u_{jt}u_{j\tau}\right) E\left(u_{js}u_{jw} \right)\right. \\
    &&\left.  - E \left(u_{jt}u_{jw} \right) E(u_{js}u_{j\tau}) \right| \\
    && + \frac{1}{K_j^\ast} \sum_{\ell=1}^{K_j^\ast}  \frac{1}{T_j^2} \sum_{t=1}^{T_j}\sum_{s=1}^{T_j}\sum_{\tau=1}^{T_j}\sum_{w=1}^{T_j} \left|  E \left(u_{jt}u_{js} \right) E \left(u_{j\tau}u_{jw} \right) + E \left(u_{jt}u_{j\tau}\right) E\left(u_{js}u_{jw} \right) \right. \\
    && \left.+ E \left(u_{jt}u_{jw} \right) E(u_{js}u_{j\tau}) \right| \\
    &=& O(1).
\end{eqnarray*}
This completes the proof.
\end{proof}

\begin{proof}[Proof of Theorem \protect\ref{Theorem: Bootstrap test}]
By Theorem \ref{Theorem: Bootstrap consistency} and Lemma \ref{lem:consistency-omegahat},
\[
    \left(\bar{Y}_1^\ast - \bar{Y}_2^\ast - \Delta_0\right)^\prime \left( \frac{\hat{\Omega}_{\text{boot},T_1}}{T_1} + \frac{\hat{\Omega}_{\text{boot},T_2}}{T_2}\right)^{-1} \left(\bar{Y}_1^\ast - \bar{Y}_2^\ast - \Delta_0\right) \rightarrow^{d^\ast} \chi^2(p),
\]
in probability. Hence, to prove Theorem \ref{Theorem: Bootstrap test}, it suffices to show $\hat{\Omega}_j^\ast - \hat{\Omega}_{\text{boot},T_j} = o_{P^\ast}(1),$ in probability. As in the proof of Lemma \ref{lem:consistency-omegahat}, we consider the case that $\hat{\Omega}_j^\ast$ and $\hat{\Omega}_{\text{boot},T_j}$ are scalars without loss of generality. 

Note that
\begin{eqnarray}
    \hat{\Omega}_j^\ast - \hat{\Omega}_{\text{boot},T_j} &=& \left(\hat{\Omega}_j^\ast - \hat{\Omega}_j \right) + \left( \hat{\Omega}_j - \hat{\Omega}_{\text{boot},T_j} \right) \notag \\
    &=& \hat{\Omega}_j^\ast - \hat{\Omega}_j  + o_P(1) \label{eqn: Pf of Thm4}
\end{eqnarray}
by consistency of $\hat{\Omega}_j$ in Theorem \ref{Theorem: Asymptotic Normality} and bootstrap consistency of $\hat{\Omega}_{\text{boot},T_j}$ in Lemma \ref{lem:consistency-omegahat}. Thus, we need to show $\hat{\Omega}_j^\ast - \hat{\Omega}_j = o_{P^\ast}(1),$ in probability.

Note that
\begin{eqnarray*}
    &&\hat{\Omega}_j^\ast - \hat{\Omega}_j \\
    &=& \frac{1}{T_j} \sum_{t=1}^{T_j}\sum_{s=1}^{T_j} \frac{1}{K_j} \sum_{\ell=1}^{K_j} \phi_\ell\left( \frac{t}{T_j} \right)\phi_\ell\left( \frac{s}{T_j} \right) \left( \hat{u}_{jt}^\ast\hat{u}_{js}^\ast  - \hat{u}_{jt}\hat{u}_{js} \right) \\
    &=& \frac{1}{T_j} \sum_{t=1}^{T_j}\sum_{s=1}^{T_j} \frac{1}{K_j} \sum_{\ell=1}^{K_j} \phi_\ell\left( \frac{t}{T_j} \right)\phi_\ell\left( \frac{s}{T_j} \right) \left(u_{jt}^\ast u_{js}^\ast- u_{jt} u_{js}  \right) + o_P(1) \\
    &=& \frac{1}{T_j} \sum_{t=1}^{T_j}\sum_{s=1}^{T_j} \frac{1}{K_j} \sum_{\ell=1}^{K_j} \phi_\ell\left( \frac{t}{T_j} \right)\phi_\ell\left( \frac{s}{T_j} \right) \left(\left({u}_{jt} - \bar{u}_j \right) \left(u_{js} - \bar{u}_j \right)\eta_{jt}\eta_{js}- u_{jt} u_{js}  \right) + o_P(1) \\
    &=& \frac{1}{T_j} \sum_{t=1}^{T_j}\sum_{s=1}^{T_j} \frac{1}{K_j} \sum_{\ell=1}^{K_j} \phi_\ell\left( \frac{t}{T_j} \right)\phi_\ell\left( \frac{s}{T_j} \right) \left[u_{jt}u_{js} \left(\eta_{jt}\eta_{js} - 1\right)- 2 \bar{u}_{j} u_{js}\eta_{jt}\eta_{js} + \bar{u}_{j}^2\eta_{jt}\eta_{js}  \right] + o_P(1) \\  
    &=& D_1 - D_2 + D_3 + o_P(1).
\end{eqnarray*}

For $D_1$,
\begin{eqnarray*}
    D_1 &=& \frac{1}{T_j} \sum_{t=1}^{T_j}\sum_{s=1}^{T_j} \frac{1}{K_j} \sum_{\ell=1}^{K_j} \phi_\ell\left( \frac{t}{T_j} \right)\phi_\ell\left( \frac{s}{T_j} \right)u_{jt}u_{js} \left(\eta_{jt}\eta_{js} - 1\right) \\
    &=& \frac{1}{T_j} \sum_{t=1}^{T_j}\sum_{s=1}^{T_j} \frac{1}{K_j} \sum_{\ell=1}^{K_j} \phi_\ell\left( \frac{t}{T_j} \right)\phi_\ell\left( \frac{s}{T_j} \right)u_{jt}u_{js} \left(\eta_{jt}\eta_{js} - E^\ast\left(\eta_{jt}\eta_{js} \right)\right) \\
    &&+ \frac{1}{T_j} \sum_{t=1}^{T_j}\sum_{s=1}^{T_j} \frac{1}{K_j} \sum_{\ell=1}^{K_j} \phi_\ell\left( \frac{t}{T_j} \right)\phi_\ell\left( \frac{s}{T_j} \right)u_{jt}u_{js} \left( E^\ast\left(\eta_{jt}\eta_{js} \right)-1\right) \\
    &=& D_{11} + D_{12}.
\end{eqnarray*}
Since $E^\ast D_{11} =0$, it suffices to show $Var^\ast \left(D_{11}\right) = o_P(1)$ to prove $D_{11} = o_{P^\ast}(1)$. Using Markov's inequality, we can prove this by showing $E\left(Var^\ast \left(D_{11}\right) \right) = o(1).$ Note that
\begin{eqnarray} \label{eqn: Var D11}
    &&E\left(Var^\ast \left(D_{11}\right)\right)  \notag \\
    &=& E\left\{E^\ast \left( \frac{1}{T_j} \sum_{t=1}^{T_j}\sum_{s=1}^{T_j} \frac{1}{K_j} \sum_{\ell=1}^{K_j} \phi_\ell\left( \frac{t}{T_j} \right)\phi_\ell\left( \frac{s}{T_j} \right)u_{jt}u_{js} \left(\eta_{jt}\eta_{js} - E^\ast\left(\eta_{jt}\eta_{js} \right)\right) \right)^2\right\} \notag \\
    &=& \frac{1}{T_j^2} \sum_{t=1}^{T_j}\sum_{s=1}^{T_j}  \sum_{\tau=1}^{T_j}\sum_{w=1}^{T_j}  E\left(u_{jt}u_{js}u_{j\tau}u_{jw}  \right)\frac{1}{K_j^2}\sum_{\ell=1}^{K_j}\sum_{k=1}^{K_j}\phi_\ell\left( \frac{t}{T_j} \right)\phi_\ell\left( \frac{s}{T_j} \right)\phi_k\left( \frac{\tau}{T_j} \right)\phi_k\left( \frac{w}{T_j} \right) \notag \\
    &&\times\left[ E^\ast \left(\eta_{jt}\eta_{js}\eta_{j\tau}\eta_{jw} \right) -  E^\ast \left(\eta_{jt}\eta_{js}\right) E^\ast \left(\eta_{j\tau}\eta_{jw} \right) \right].
\end{eqnarray}

Note that
\begin{eqnarray} \label{eqn: D11 decomposition}
&&E^\ast \left(\eta_{jt}\eta_{js}\eta_{j\tau}\eta_{jw} \right) - E^\ast \left(\eta_{jt}\eta_{js}\right) E^\ast \left(\eta_{j\tau}\eta_{jw} \right) \notag \\ 
&=& \sum_{r=1}^2\frac{1}{K_j^{\ast 2}} \sum_{\ell=1}^{K_j^\ast}  \psi_{r,\ell}\left( \frac{t}{T_j} \right)\psi_{r,\ell}\left( \frac{s}{T_j} \right)\psi_{r,\ell}\left( \frac{\tau}{T_j} \right)\psi_{r,\ell}\left( \frac{w}{T_j} \right) \left(E^\ast\left(v_{j,r\ell}^4 \right)  - 3 \right) \notag \\
&&+\left(\sum_{r=1}^2 \frac{1}{K_j^\ast} \sum_{\ell=1}^{K_j^\ast} \psi_{r,\ell}\left( \frac{t}{T_j} \right) \psi_{r,\ell}\left( \frac{\tau}{T_j} \right)  \right) \left(\sum_{a=1}^2 \frac{1}{K_j^\ast}\sum_{k=1}^{K_j^\ast} \psi_{a,k}\left( \frac{s}{T_j} \right)\psi_{a,k}\left( \frac{w}{T_j} \right) \right) \notag \\
&&+ \left(\sum_{r=1}^2 \frac{1}{K_j^\ast} \sum_{\ell=1}^{K_j^\ast} \psi_{r,\ell}\left( \frac{t}{T_j} \right) \psi_{r,\ell}\left( \frac{w}{T_j} \right)  \right) \left( \sum_{a=1}^2  \frac{1}{K_j^\ast}\sum_{k=1}^{K_j^\ast} \psi_{a,k}\left( \frac{s}{T_j} \right)\psi_{a,k}\left( \frac{\tau}{T_j} \right) \right).
\end{eqnarray}
Plugging (\ref{eqn: D11 decomposition}) into (\ref{eqn: Var D11}), we have
\begin{eqnarray*} 
    &&E\left(Var^\ast \left(D_{11}\right)\right)  \notag \\
    &=& \frac{1}{T_j^2} \sum_{t=1}^{T_j}\sum_{s=1}^{T_j}  \sum_{\tau=1}^{T_j}\sum_{w=1}^{T_j}  E\left(u_{jt}u_{js}u_{j\tau}u_{jw}  \right)\frac{1}{K_j^2}\sum_{\ell=1}^{K_j}\sum_{k=1}^{K_j}\phi_\ell\left( \frac{t}{T_j} \right)\phi_\ell\left( \frac{s}{T_j} \right)\phi_k\left( \frac{\tau}{T_j} \right)\phi_k\left( \frac{w}{T_j} \right) \notag \\
    &&\times\left[ \sum_{r=1}^2\frac{1}{K_j^{\ast 2}} \sum_{\ell=1}^{K_j^\ast}\psi_{r,\ell}\left( \frac{t}{T_j} \right)\psi_{r,\ell}\left( \frac{s}{T_j} \right)\psi_{r,\ell}\left( \frac{\tau}{T_j} \right)\psi_{r,\ell}\left( \frac{w}{T_j} \right) \left(E^\ast\left(v_{j,r\ell}^4 \right)  - 3 \right) \notag  \right.\\
&&+\left(\sum_{r=1}^2 \frac{1}{K_j^\ast} \sum_{\ell=1}^{K_j^\ast} \psi_{r,\ell}\left( \frac{t}{T_j} \right) \psi_{r,\ell}\left( \frac{\tau}{T_j} \right)  \right) \left(\sum_{a=1}^2 \frac{1}{K_j^\ast}\sum_{k=1}^{K_j^\ast} \psi_{a,k}\left( \frac{s}{T_j} \right)\psi_{a,k}\left( \frac{w}{T_j} \right) \right) \notag \\
&& \left. + \left(\sum_{r=1}^2 \frac{1}{K_j^\ast} \sum_{\ell=1}^{K_j^\ast} \psi_{r,\ell}\left( \frac{t}{T_j} \right) \psi_{r,\ell}\left( \frac{w}{T_j} \right)  \right) \left( \sum_{a=1}^2  \frac{1}{K_j^\ast}\sum_{k=1}^{K_j^\ast} \psi_{a,k}\left( \frac{s}{T_j} \right)\psi_{a,k}\left( \frac{\tau}{T_j} \right) \right)\right] \\
&=& d_1 + d_2 + d_3.
\end{eqnarray*}
For $d_1$,
\begin{eqnarray*}
    \left| d_1 \right| &\leq& O\left(\frac{1}{K_j^\ast}\right) \frac{1}{T_j^2} \sum_{t=1}^{T_j}\sum_{s=1}^{T_j}  \sum_{\tau=1}^{T_j}\sum_{w=1}^{T_j} \frac{1}{K_j^2}\sum_{\ell=1}^{K_j}\sum_{k=1}^{K_j} \left| \phi_\ell\left( \frac{t}{T_j} \right)\phi_\ell\left( \frac{s}{T_j} \right)\phi_k\left( \frac{\tau}{T_j} \right)\phi_k\left( \frac{w}{T_j} \right) \right| \\
    &&\times \sum_{r=1}^2  \frac{1}{K_j^\ast} \sum_{\ell=1}^{K_j^\ast} \left| \psi_{r,\ell}\left( \frac{t}{T_j} \right)\psi_{r,\ell}\left( \frac{s}{T_j} \right)\psi_{r,\ell}\left( \frac{\tau}{T_j} \right)\psi_{r,\ell}\left( \frac{w}{T_j} \right) \right| \underbrace{\left|E^\ast\left(v_{j,r\ell}^4 \right)  - 3 \right|}_{<M} \\
    &=& O\left(\frac{1}{K_j^\ast}\right).
\end{eqnarray*}


For $d_2$,
\begin{eqnarray*}
    d_2    &=&  \frac{1}{T_j^2} \sum_{t=1}^{T_j}\sum_{s=1}^{T_j}  \sum_{\tau=1}^{T_j}\sum_{w=1}^{T_j} \left[ E\left(u_{jt}u_{js}u_{j\tau}u_{jw}  \right) -E\left(u_{jt}u_{js}\right) E\left(u_{j\tau}u_{jw}\right) - E\left(u_{jt}u_{j\tau} \right) E\left(u_{js}u_{jw} \right) \right.\\
    && \left. - E\left(u_{jt}u_{jw} \right) E\left(u_{js}u_{j\tau}  \right) \right] \left(\frac{1}{K_j}\sum_{\ell=1}^{K_j}\phi_\ell\left( \frac{t}{T_j} \right)\phi_\ell\left( \frac{s}{T_j} \right) \right)\left( \frac{1}{K_j}\sum_{k=1}^{K_j}\phi_k\left( \frac{\tau}{T_j} \right)\phi_k\left( \frac{w}{T_j} \right)\right) \\
    &&\times \left(\sum_{r=1}^2 \frac{1}{K_j^\ast} \sum_{\ell_1=1}^{K_j^\ast} \psi_{r,\ell_1}\left( \frac{t}{T_j} \right) \psi_{r,\ell_1}\left( \frac{\tau}{T_j} \right)  \right) \left(\sum_{r=1}^2 \frac{1}{K_j^\ast}\sum_{k_1=1}^{K_j^\ast} \psi_{a,k_1}\left( \frac{s}{T_j} \right)\psi_{a,k_1}\left( \frac{w}{T_j} \right) \right) \notag \\
    && + \frac{1}{T_j^2} \sum_{t=1}^{T_j}\sum_{s=1}^{T_j}  \sum_{\tau=1}^{T_j}\sum_{w=1}^{T_j} \left[E\left(u_{jt}u_{js}\right) E\left(u_{j\tau}u_{jw}\right) + E\left(u_{jt}u_{j\tau} \right) E\left(u_{js}u_{jw} \right) +E\left(u_{jt}u_{jw} \right) E\left(u_{js}u_{j\tau}  \right) \right] \\
    && \times \left(\frac{1}{K_j}\sum_{\ell=1}^{K_j}\phi_\ell\left( \frac{t}{T_j} \right)\phi_\ell\left( \frac{s}{T_j} \right) \right)\left( \frac{1}{K_j}\sum_{k=1}^{K_j}\phi_k\left( \frac{\tau}{T_j} \right)\phi_k\left( \frac{w}{T_j} \right)\right)\\
    && \times\left(\sum_{r=1}^2 \frac{1}{K_j^\ast} \sum_{\ell=1}^{K_j^\ast} \psi_{r,\ell}\left( \frac{t}{T_j} \right) \psi_{r,\ell}\left( \frac{\tau}{T_j} \right)  \right) \left(\sum_{a=1}^2 \frac{1}{K_j^\ast}\sum_{k=1}^{K_j^\ast} \psi_{a,k}\left( \frac{s}{T_j} \right)\psi_{a,k}\left( \frac{w}{T_j} \right) \right)\\
    &=& d_{21} + d_{22} + d_{23} + d_{24}.
\end{eqnarray*}
For $d_{21}$, it is easy to show 
\[
d_{21} = O\left(\frac{1}{T_j}\right)
\]
using Assumption \ref{Assumption3: cumulant}. For $d_{22}$,
\begin{eqnarray*}
    d_{22} &=& \sum_{r=1}^2 \sum_{a=1}^2  \frac{1}{K_j^{\ast 2}} \sum_{\ell_1=1}^{K_j^\ast}\sum_{k_1=1}^{K_j^\ast} \left[\frac{1}{T_j} \sum_{t=1}^{T_j}\sum_{s=1}^{T_j}   E\left(u_{jt}u_{js}\right) \left( \frac{1}{K_j}\sum_{\ell=1}^{K_j} \phi_\ell\left( \frac{t}{T_j} \right)\phi_\ell\left( \frac{s}{T_j} \right) \right)  \right. \\
    &&\times \left. \psi_{r,\ell_1}\left( \frac{t}{T_j} \right)\psi_{a,k_1}\left( \frac{s}{T_j} \right) \right]^2. \\ 
\end{eqnarray*}
Since
\begin{eqnarray*}
    &&\left[ \frac{1}{T_j} \sum_{t=1}^{T_j}\sum_{s=1}^{T_j}   E\left(u_{jt}u_{js}\right) \left( \frac{1}{K_j}\sum_{\ell=1}^{K_j} \phi_\ell\left( \frac{t}{T_j} \right)\phi_\ell\left( \frac{s}{T_j} \right) \right) \psi_{r,\ell_1}\left( \frac{t}{T_j} \right)\psi_{a,k_1}\left( \frac{s}{T_j} \right)\right]^2 \\
    &\leq& 2 \left[ \frac{1}{T_j} \sum_{t=1}^{T_j}\sum_{s=1}^{T_j}   E\left(u_{jt}u_{js}\right) \left( \frac{1}{K_j}\sum_{\ell=1}^{K_j} \phi_\ell\left( \frac{t}{T_j} \right)\phi_\ell\left( \frac{s}{T_j} \right) -1\right) \psi_{r,\ell_1}\left( \frac{t}{T_j} \right)\psi_{a,k_1}\left( \frac{s}{T_j} \right)\right]^2 \\
    &&+2 \left[ \frac{1}{T_j} \sum_{t=1}^{T_j}\sum_{s=1}^{T_j}   E\left(u_{jt}u_{js}\right)  \psi_{r,\ell_1}\left( \frac{t}{T_j} \right)\psi_{a,k_1}\left( \frac{s}{T_j} \right)\right]^2,
\end{eqnarray*}
it suffices to show that
\begin{eqnarray*}
     d_{221}&=&\sum_{r=1}^2 \sum_{a=1}^2  \frac{1}{K_j^{\ast 2}} \sum_{\ell_1=1}^{K_j^\ast}\sum_{k_1=1}^{K_j^\ast} \left[ \frac{1}{T_j} \sum_{t=1}^{T_j}\sum_{s=1}^{T_j}   E\left(u_{jt}u_{js}\right) \left( \frac{1}{K_j}\sum_{\ell=1}^{K_j} \phi_\ell\left( \frac{t}{T_j} \right)\phi_\ell\left( \frac{s}{T_j} \right) -1\right)  \right. \\
     && \times \left.\psi_{r,\ell_1}\left( \frac{t}{T_j} \right)\psi_{a,k_1}\left( \frac{s}{T_j} \right)\right]^2 = o(1)
\end{eqnarray*}
and 
\begin{eqnarray*}
     d_{222}= \sum_{r=1}^2 \sum_{a=1}^2  \frac{1}{K_j^{\ast 2}} \sum_{\ell_1=1}^{K_j^\ast}\sum_{k_1=1}^{K_j^\ast} \left[ \frac{1}{T_j} \sum_{t=1}^{T_j}\sum_{s=1}^{T_j}   E\left(u_{jt}u_{js}\right)  \psi_{r,\ell_1}\left( \frac{t}{T_j} \right)\psi_{a,k_1}\left( \frac{s}{T_j} \right) \right]^2 = o(1).
\end{eqnarray*}

For $d_{221}$, we can show $d_{221} = o(1)$ using an argument similar to the one to study the term $B_{11}$ in the proof of Lemma \ref{lem:consistency-omegahat}. 

For $d_{222}$,
\begin{eqnarray*}
    d_{222} &=& \sum_{r=1}^2 \sum_{a=1}^2  \frac{1}{K_j^{\ast 2}} \sum_{\ell_1=1}^{K_j^\ast} \left[ \frac{1}{T_j} \sum_{t=1}^{T_j}\sum_{s=1}^{T_j}   E\left(u_{jt}u_{js}\right)  \psi_{r,\ell_1}\left( \frac{t}{T_j} \right)\psi_{a,\ell_1}\left( \frac{s}{T_j} \right) \right]^2 \\
    &&+ \sum_{r=1}^2 \sum_{a=1}^2  \frac{1}{K_j^{\ast 2}} \sum_{\ell_1=1}^{K_j^\ast}\sum_{k_1\neq \ell_1} \left[ \frac{1}{T_j} \sum_{t=1}^{T_j}\sum_{s=1}^{T_j}   E\left(u_{jt}u_{js}\right)  \psi_{r,\ell_1}\left( \frac{t}{T_j} \right)\psi_{a,k_1}\left( \frac{s}{T_j} \right) \right]^2.
\end{eqnarray*}
It is straightforward to show that the first term is $O(1/K_j^{\ast})$. Suppose that $L_{3T_j}$ satisfies $L_{3T_j} \rightarrow \infty$ as $T_j \rightarrow \infty$ such that $L_{3T_j}/T_j \rightarrow 0$. For the second term, we note that
\begin{eqnarray*}
    &&  \frac{1}{T_j} \sum_{t=1}^{T_j}\sum_{s=1}^{T_j}   E\left(u_{jt}u_{js}\right)  \psi_{r,\ell_1}\left( \frac{t}{T_j} \right)\psi_{a,k_1}\left( \frac{s}{T_j} \right)  \\
    &=& \sum_{h=-L_{3T_j}}^{L_{3T_j}}\frac{1}{T_j} \sum_{1 \leq t,t+h \leq T_j}\Gamma_j(h) \psi_{r,\ell_1}\left( \frac{t}{T_j} \right)\psi_{a,k_1}\left( \frac{t+h}{T_j} \right) \\
    &&+ \sum_{L_{3T_j}<|h|<T_j}\frac{1}{T_j} \sum_{1 \leq t,t+h \leq T_j}\Gamma_j(h) \psi_{r,\ell_1}\left( \frac{t}{T_j} \right)\psi_{a,k_1}\left( \frac{t+h}{T_j} \right) \\
    &=& d_{2221} + d_{2222}.
\end{eqnarray*}
For $d_{2222}$,
\begin{eqnarray*}
    |d_{2222}| &\leq& \frac{1}{L_{3T_j}} \sum_{L_{3T_j}<|h|<T_j} |\Gamma_j(h)| |h|  \frac{1}{T_j} \sum_{1 \leq t,t+h \leq T_j} \left| \psi_{r,\ell_1}\left( \frac{t}{T_j} \right)\psi_{a,k_1}\left( \frac{t+h}{T_j} \right) \right|\\
    &=& o(1).
\end{eqnarray*}
For $d_{2221}$,
\begin{eqnarray*}
    |d_{2221}| &=& \left|\sum_{h=-L_{3T_j}}^{L_{3T_j}} \Gamma_j(h) \frac{1}{T_j} \sum_{1 \leq t,t+h \leq T_j} \psi_{r,\ell_1}\left( \frac{t}{T_j} \right) \left( \psi_{a,k_1}\left( \frac{t}{T_j} \right) + \psi_{a,k_1}^\prime \left( \frac{t+\tilde{h}}{T_j} \right) \frac{h}{T_j} \right) \right| \\
    &=& \sum_{h=-L_{3T_j}}^{L_{3T_j}} | \Gamma_j(h) | \left| \frac{1}{T_j} \sum_{1 \leq t,t+h \leq T_j} \psi_{r,\ell_1}\left( \frac{t}{T_j} \right)\psi_{a,k_1}\left( \frac{t}{T_j} \right) \right| \\
    &&+O\left(\frac{1}{T_j} \right)\sum_{h=-L_{3T_j}}^{L_{3T_j}} |\Gamma_j(h)| |h| \left|\frac{1}{T_j} \sum_{1 \leq t,t+h \leq T_j} \psi_{r,\ell_1}\left( \frac{t}{T_j} \right)\psi_{a,k_1}^\prime \left( \frac{t+\tilde{h}}{T_j} \right) \right| \\
    &=& o(1).
\end{eqnarray*}
For the first term in the second equation, let's focus on the case that $h\geq0$. Since $\frac{L_{3T_j}}{T_j} \rightarrow 0$ and $|h| \leq L_{3T_j}$,
\[
\frac{1}{T_j} \sum_{1 \leq t,t+h \leq T_j} \psi_{r,\ell_1}\left( \frac{t}{T_j} \right)\psi_{a,k_1}\left( \frac{t}{T_j} \right) = \int_{0}^{1-\frac{h}{T_j}} \psi_{r,\ell_1}\left( x \right)\psi_{a,k_1}\left( x\right) dx = o(1)
\]
by Lemma \ref{Lemma: Basis function}.

For $D_3$,
\begin{eqnarray*}
    D_3 &=& \left(\sqrt{T_j} \bar{u}_{j} \right)^2\frac{1}{K_j} \sum_{\ell=1}^{K_j} \left(\frac{1}{T_j} \sum_{t=1}^{T_j}\phi_\ell\left( \frac{t}{T_j} \right)\eta_{jt}\right)\left(\frac{1}{T_j} \sum_{s=1}^{T_j}\phi_\ell\left( \frac{s}{T_j} \right)\eta_{js}\right)\\
    &\leq& \left(\sqrt{T_j} \bar{u}_{j} \right)^2 \underbrace{\sqrt{\frac{1}{K_j} \sum_{\ell=1}^{K_j} \left(\frac{1}{T_j} \sum_{t=1}^{T_j}\phi_\ell\left( \frac{t}{T_j} \right)\eta_{jt}\right)^2}}_{:=e_1 = o_{P^\ast}(1) \text{ in probability}} \sqrt{\frac{1}{K_j} \sum_{\ell=1}^{K_j} \left(\frac{1}{T_j} \sum_{s=1}^{T_j}\phi_\ell\left( \frac{s}{T_j} \right)\eta_{js}\right)^2} \\
    &=& o_{P^\ast} \left(1 \right),
\end{eqnarray*}
in probability. To show $e_1 = o_{P^\ast}(1),$ in probability, it suffices to show that $E\left(E^\ast\left( e_1^2 \right)\right) = o(1)$.  
\begin{eqnarray*}
    E^\ast\left( e_1^2 \right) &=& \frac{1}{K_j}\sum_{\ell=1}^{K_j} \frac{1}{T_j^2}\sum_{t=1}^{T_j}\sum_{s=1}^{T_j} \phi_{\ell}\left( \frac{t}{T_j} \right)\phi_{\ell}\left( \frac{s}{T_j}\right) E^\ast \left(\eta_{jt}\eta_{js} \right) \\
    &=& \frac{1}{K_j K_j^\ast}\sum_{\ell=1}^{K_j}\sum_{\ell_1=1}^{K_j^\ast}  \left[\left(\frac{1}{T_j}\sum_{t=1}^{T_j} \phi_{\ell}\left( \frac{t}{T_j} \right) \psi_{1,\ell_1}\left( \frac{t}{T_j} \right)\right)^2  + \left(\frac{1}{T_j}\sum_{t=1}^{T_j} \phi_{\ell}\left( \frac{t}{T_j} \right)\psi_{2,\ell_1}\left( \frac{t}{T_j} \right)\right)^2 \right]  \\
    &=& e_{11} + e_{12}.
\end{eqnarray*}
Assume $K_j$ is a even number for simplicity. Since 
\[
\phi_{2\ell-1}(x) = \sqrt{2} \psi_{1,\ell}(x), \qquad \phi_{2\ell}(x) = \sqrt{2} \psi_{2,\ell}(x),
\]
\begin{eqnarray*}
    e_{11} &=& \frac{1}{K_j K_j^\ast}\sum_{\ell=1}^{K_j/2}\sum_{\ell_1=1}^{K_j^\ast} \left[\left(\frac{1}{T_j}\sum_{t=1}^{T_j} \phi_{2\ell-1}\left( \frac{t}{T_j} \right) \psi_{1,\ell_1}\left( \frac{t}{T_j} \right)\right)^2 + \left(\frac{1}{T_j}\sum_{t=1}^{T_j} \phi_{2\ell}\left( \frac{t}{T_j} \right) \psi_{1,\ell_1}\left( \frac{t}{T_j} \right)\right)^2 \right] \\
    &=& \frac{2}{K_j K_j^\ast} \sum_{\ell_1=1}^{\min( K_j/2,K_j^\ast)} \left(\frac{1}{T_j}\sum_{t=1}^{T_j} \psi_{1,\ell_1}^2\left( \frac{t}{T_j} \right)  \right)^2 + o(1) \\
    &=& o\left(1 \right)
\end{eqnarray*}
with $K_j, K_j^\ast \rightarrow \infty$, where we use the orthogonality of $\{ \psi_{r,\ell}(x) \}$ in Lemma \ref{Lemma: Basis function}. We can use the same argument to prove $e_{12} = o\left(1\right)$. Thus, we have $E\left(E^\ast\left(e_1^2 \right)\right) =o\left(1\right)$.

For $D_2,$
\begin{eqnarray*}
    D_2 &=& 2\sqrt{T_j}\bar{u}_j \frac{1}{K_j} \sum_{\ell=1}^{K_j}\left(\frac{1}{T_j} \sum_{t=1}^{T_j} \phi_\ell\left( \frac{t}{T_j} \right)\eta_{jt} \right)\left(\frac{1}{\sqrt{T_j}} \sum_{s=1}^{T_j}  \phi_\ell\left( \frac{s}{T_j} \right) u_{js}\eta_{js}\right) \\
    &\leq& 2\sqrt{T_j}\bar{u}_j \underbrace{\sqrt{\frac{1}{K_j} \sum_{\ell=1}^{K_j}\left(\frac{1}{T_j} \sum_{t=1}^{T_j} \phi_\ell\left( \frac{t}{T_j} \right)\eta_{jt} \right)^2}}_{=o_{P^\ast}(1), \text{ in probability, as shown above}} \underbrace{\sqrt{\frac{1}{K_j} \sum_{\ell=1}^{K_j}\left(\frac{1}{\sqrt{T_j}} \sum_{s=1}^{T_j}  \phi_\ell\left( \frac{s}{T_j} \right) u_{js}\eta_{js}\right)^2 }}_{:=d_2 = O_{P^\ast}(1)} \\
    &=& o_{P^\ast}(1),
\end{eqnarray*}
where we have $d_2 = O_{P^\ast}(1)$ because
\begin{eqnarray*}
    E\left( E^\ast\left(d_2^2 \right) \right) 
    &\leq& O(1) \frac{1}{K_j} \sum_{\ell=1}^{K_j}\frac{1}{T_j} \sum_{t=1}^{T_j} \sum_{s=1}^{T_j}  \phi_\ell\left( \frac{t}{T_j} \right)\phi_\ell\left( \frac{s}{T_j} \right) \left \vert E\left(u_{jt}u_{js}\right) \right \vert \\
    &=& O(1).
\end{eqnarray*}
Therefore, 
\[
\hat{\Omega}_j^\ast - \hat{\Omega}_j = o_{P^\ast}(1),
\]
in probability, which completes the proof as shown in (\ref{eqn: Pf of Thm4}).
\end{proof}

\begin{proof}[Proof of Lemma \protect\ref{Lemma: Basis function}]
Consider the case of $r=c=1$. Then, we have $\psi _{1,\ell}\left( x\right) = \cos \left( 2 \pi \ell x \right)$. For (i),
\begin{eqnarray*}
\frac{1}{T}\sum_{t=1}^{T}\psi _{1,\ell}\left( \frac{t}{T}\right)  &=& \int_0^1 \cos(2\pi \ell x) dx + o(1) = o(1).
\end{eqnarray*}%
For (ii), if $\ell=k$,
\begin{eqnarray*}
\frac{1}{T}\sum_{t=1}^{T}\psi^2_{1,\ell}\left(\frac{t}{T}\right) &=& \int_0^1 \cos^2 \left(2 \pi \ell x \right) dx (1+o(1)) \\
&=& \frac{1}{2}\int_0^1 \left( 1 + \cos\left( 4\pi \ell x\right) \right)dx(1+o(1))\\
&=& \frac{1}{2} + o(1).
\end{eqnarray*}%
If $\ell \neq k$
\begin{eqnarray*}
\frac{1}{T}\sum_{t=1}^{T}\psi _{1,\ell}\left(\frac{t}{T}\right) \psi_{1,k}\left(\frac{t}{T}\right)
&=&\int_0^1 \cos(2\pi \ell x) \cos(2\pi k x) dx(1 + o(1))\\
&=& \frac{1}{2} \int_0^1 \left[ \cos(2\pi(\ell-k)x) + \cos(2\pi(\ell+k)x) \right]dx(1+o(1)) \\
&=& o(1).
\end{eqnarray*}%
Using the same procedure, we can also show that (i) and (ii) hold in the other cases.
\end{proof}

\begin{proof}[Proof of Lemma \protect\ref{lem:consistency-omegahat}] Note that $\hat{\Omega}_{\text{boot},T_j} - \Omega_j \rightarrow^P 0$ if and only if $\underline{g}^\prime \hat{\Omega}_{\text{boot},T_j} \underline{g}  - \underline{g}^\prime \Omega_j \underline{g} \rightarrow^P 0 $ for any $p \times 1$ vector $\underline{g}$. Therefore, we assume the case that $\hat{\Omega}_{\text{boot},T_j}$ and $\Omega_j$ are scalars without loss of generality.
Define
\begin{eqnarray*}
\Omega_{T_j} &=& Var\left( \frac{1}{\sqrt{T_j}} \sum_{t=1}^{T_j} u_{jt}\right),\qquad
\Omega_{\text{boot},T_j} = Var^\ast\left( \frac{1}{\sqrt{T_j}} \sum_{t=1}^{T_j} u_{jt}\eta_{jt} \right).
\end{eqnarray*}
$\Omega_{T_j}$ is the finite sample version of the LRV, and $\Omega_{\text{boot},T_j}$ is the infeasible version of the bootstrap LRV by replacing $\hat{u}_{jt}$ in $\hat{\Omega}_{\text{boot},T_j}$ with $u_{jt}$.

The proof consists of two parts: 
\[
\text{(i)} \ \hat{\Omega}_{\text{boot},T_j} - \Omega_{\text{boot},T_j} = o_P(1), \qquad
\text{(ii)} \ \Omega_{\text{boot},T_j} - \Omega_{T_j} = o_P(1).
\]
For part (i), note that%
\begin{eqnarray*}
&&\hat{\Omega}_{\text{boot},T_j}-\Omega_{\text{boot},T_j} \\
&=&\frac{1}{T_j}\sum_{t=1}^{T_j}\sum_{s=1}^{T_j}\frac{1}{K_j^\ast}%
\sum_{\ell=1}^{K_j^\ast} \left[\psi_{1,\ell}\left( \frac{t}{T_j}\right) \psi_{1,\ell}\left(\frac{s}{T_j}\right) + \psi_{2,\ell}\left( \frac{t}{T_j}\right) \psi_{2,\ell}\left(\frac{s}{T_j}\right) \right]\left( \hat{u}_{jt}\hat{u}_{js}-u_{jt}u_{js}\right)  \\
&=&\frac{1}{T_j}\sum_{t=1}^{T_j}\sum_{s=1}^{T_j}\frac{1}{K_j^\ast}%
\sum_{\ell=1}^{K_j^\ast}\left[\psi_{1,\ell}\left( \frac{t}{T_j}\right) \psi_{1,\ell}\left(\frac{s}{T_j}\right) + \psi_{2,\ell}\left( \frac{t}{T_j}\right) \psi_{2,\ell}\left(\frac{s}{T_j}\right) \right]\left(\bar{u}_j \bar{u}_j - u_{jt}\bar{u}_j - \bar{u}_ju_{js} \right) \\
&:=& A_1 + A_2 + A_3.
\end{eqnarray*}%
For $A_1$,
\begin{eqnarray} \label{eqn: A1}
    A_1 = \left( \sqrt{T_j} \bar{u}_j \right)^2\ \frac{1}{K_j^\ast}\sum_{\ell=1}^{K_j^\ast} \sum_{r=1}^{2} \underbrace{\frac{1}{T_j}\sum_{t=1}^{T_j} \psi _{r,\ell}\left( \frac{t}{T_j}\right) }_{=o(1) \text{ by Lemma \ref{Lemma: Basis function}}}\frac{1}{T_j}\sum_{s=1}^{T_j} \psi_{r,\ell}\left( \frac{s}{T_j} \right)= o_P(1).
\end{eqnarray}
Using a similar procedure as in (\ref{eqn: A1}), we can show $A_2 = A_3 = o_P(1)$. 

For part (ii), 
\begin{eqnarray*}
\Omega_{\text{boot},T_j}-\Omega_{T_j}  
    &=& \frac{1}{T_j} \sum_{t=1}^{T_j} \sum_{s=1}^{T_j}  \frac{1}{K_j^\ast} \sum_{\ell=1}^{K_{j}^\ast} \psi_{1,\ell}\left( \frac{t}{T_j} \right)\psi_{1,\ell}\left( \frac{s}{T_j} \right)  u_{jt} u_{js} - \frac{1}{2}\Omega_{T_j} \\
        &&+ \frac{1}{T_j} \sum_{t=1}^{T_j} \sum_{s=1}^{T_j}  \frac{1}{K_j^\ast} \sum_{\ell=1}^{K_{j}^\ast} \psi_{2,\ell}\left( \frac{t}{T_j} \right)\psi_{2,\ell}\left( \frac{s}{T_j} \right)  u_{jt} u_{js} - \frac{1}{2}\Omega_{T_j}\\
        &=& B_1 + B_2.
\end{eqnarray*}
To show $B_1=o_P(1),$ we first decompose $B_1$ into the bias
and variation, and then show that each term converges to zero. 
\begin{eqnarray*}
B_1 &=&  \frac{1%
}{2K_j^\ast}\sum_{\ell=1}^{K_j^\ast} \frac{1}{T_j}\sum_{t=1}^{T_j}\sum_{s=1}^{T_j}\left[2\psi _{1,\ell}\left( \frac{t}{T_j}\right) \psi
_{1,\ell}\left( \frac{s}{T_j}\right) -1\right] Eu_{jt}u_{js} \\
&&+\frac{1}{K_j^\ast}%
\sum_{\ell=1}^{K_j^\ast}\frac{1}{T_j}\sum_{t=1}^{T_j}\sum_{s=1}^{T_j}\psi _{1,\ell}\left( \frac{t}{T_j}\right) \psi _{1,\ell}\left( 
\frac{s}{T_j}\right) \left[ u_{jt}u_{js}-Eu_{jt}u_{js}\right]  \\
&=&B_{11}+B_{12}.
\end{eqnarray*}%
Recall $\Gamma_{j}(t-s) = Eu_{jt}u_{js}^\prime$. Choose $L_{1T_j}$ that satisfies
\begin{equation*}
\frac{T_j}{L_{1T_j}^{3/2}K_j^\ast}+\frac{L_{1T_j}K_j^\ast}{T_j}=o\left( 1\right) 
\end{equation*}%
as $T_j,K_j \rightarrow \infty$ and $K_j/T_j \rightarrow 0.$
\begin{eqnarray}
B_{11} &=&\frac{1}{2K_j^\ast}\sum_{\ell=1}^{K_j^\ast}\sum_{h=-L_{1T_j}}^{L_{1T_j}}%
\left[ \frac{1}{T_j}\sum_{1\leq t,t+h\leq T_j}2\psi _{1,\ell}\left( \frac{t+h}{T_j}\right)
\psi_{1,\ell}\left( \frac{t}{T_j}\right) -1\right] \Gamma_j \left( h\right)  
\notag \\
&&+\frac{1}{2K_j^\ast}\sum_{\ell=1}^{K_j^\ast}\sum_{L_{1T_j}<\left\vert
h\right\vert <T_j}\left[ \frac{1}{T_j}%
\sum_{1\leq t,t+h\leq T_j}2\psi _{1,\ell}\left( \frac{t+h}{T_j}\right) \psi _{1,\ell}\left( 
\frac{t}{T_j}\right) - 1\right] \Gamma_j \left( h\right)   \notag \\
&=&b_{1}+b_{2}.  \label{eqn: G1}
\end{eqnarray}
For $b_2$,
\begin{equation*}
\left\Vert b_{2}\right\Vert \leq O\left( 1\right) \frac{1}{K_j^\ast}%
\sum_{\ell=1}^{K_j^\ast}\sum_{h=-T_j}^{T_j}\left\vert
\Gamma_j \left( h\right) \right\vert \left\vert h\right\vert ^{3}\frac{1}{%
L_{1T_j}^{3}}=O\left( \frac{1}{L_{1T_j}^{3}}\right) =o\left( \frac{K_j^{\ast2}}{%
T_j^{2}}\right). 
\end{equation*}%
For $b_{1}$, 
\begin{eqnarray*}
b_{1} &=&\frac{1}{2K_j^\ast}\sum_{\ell=1}^{K_j^\ast}\sum_{h=-L_{1T_j}}^{L_{1T_j}}%
\left[ \frac{1}{T_j}\sum_{1\leq t,t+h\leq T_j}2\psi _{1,\ell}\left( \frac{t}{T_j}\right)
\left( \psi _{1,\ell}\left( \frac{t+h}{T_j}\right) - \psi _{1,\ell}\left( \frac{t}{T_j}%
\right) \right) \right] \Gamma_j \left( h\right)  \\
&&+\frac{1}{2K_j^\ast}%
\sum_{\ell=1}^{K_j^\ast}\sum_{h=-L_{1T_j}}^{L_{1T_j}}\left[ \frac{1}{T_j}\sum_{1\leq
t,t+h\leq T_j}2\psi_{1,\ell}^2\left( \frac{t}{T_j} \right)-1\right] \Gamma_j \left(h\right)\\
&=& b_{11} + b_{12}.
\end{eqnarray*}%
For $b_{11}$, since
\begin{eqnarray*}
&& \frac{1}{T_j}\sum_{1\leq t,t+h\leq T_j}2\psi _{1,\ell}\left( \frac{t}{T_j}\right)
\left( \psi _{1,\ell}\left( \frac{t+h}{T_j}\right) - \psi _{1,\ell}\left( \frac{t}{T_j}%
\right) \right) \\
&=&\frac{1}{T_j}\sum_{1\leq t,t+h\leq T_j}2\cos\left(2 \pi \ell \frac{t}{T_j}\right)
\left( \cos\left(2 \pi \ell \frac{t+h}{T_j}\right) - \cos\left(2 \pi \ell \frac{t}{T_j}%
\right) \right) \\
&=& \frac{1}{T_j}\sum_{1\leq t,t+h\leq T_j}2\cos\left(2 \pi \ell \frac{t}{T_j}\right) \sin\left(2 \pi \ell \frac{t+\tilde{h}}{T_j}\right)\frac{2\pi\ell h}{T_j} 
\end{eqnarray*}
for some $\tilde{h} \in (0,h)$, we have
\begin{eqnarray*}
    \left\vert b_{11} \right\vert &\leq& O(1) \frac{1}{K_j^\ast}\sum_{\ell=1}^{K_j^\ast}\sum_{h=-L_{1T_j}}^{L_{1T_j}}\left\Vert\Gamma_j \left(h\right) \right\Vert \left | h \right| \frac{\ell}{T_j} = O\left(\frac{K_j^\ast}{T_j} \right).
\end{eqnarray*}
For $b_{12}$,
\begin{eqnarray*}
 \vert b_{12} \vert  &=&\left\vert \frac{1}{2K_j^\ast}%
\sum_{\ell=1}^{K_j^\ast}\sum_{h=-L_{1T_j}}^{L_{1T_j}}\left[ \frac{1}{T_j}\sum_{1\leq
t,t+h\leq T_j}2\cos^2\left( 2\pi \ell \frac{t}{T_j} \right)-1\right] \Gamma_j \left(h\right) \right \vert \\
&\leq & \left \vert \frac{1}{2K_j^\ast}%
\sum_{\ell=1}^{K_j^\ast}\sum_{h=-L_{1T_j}}^{L_{1T_j}}\left[ \frac{1}{T_j}\sum_{1\leq
t,t+h\leq T_j} \left(1 + \cos\left( 4\pi \ell \frac{t}{T_j} \right)\right)-1\right] \Gamma_j \left(h\right) \right\vert \\
&\leq& \left\vert\frac{1}{2K_j^\ast}%
\sum_{\ell=1}^{K_j^\ast}\sum_{h=-L_{1T_j}}^{L_{1T_j}}\left[ \frac{1}{T_j}\sum_{1\leq
t,t+h\leq T_j} 1 -1\right] \Gamma_j \left(h\right) \right\vert\\
&&+ \left\vert\frac{1}{2K_j^\ast}%
\sum_{\ell=1}^{K_j^\ast}\sum_{h=-L_{1T_j}}^{L_{1T_j}}\left[  \frac{1}{T_j}\sum_{1\leq
t,t+h\leq T_j} \cos\left( 4\pi \ell \frac{t}{T_j} \right)\right] \Gamma_j \left(h\right) \right\vert \\
&=& b_{121} + b_{122}.
\end{eqnarray*}    
For $b_{121}$,
\begin{eqnarray*}
b_{121} &\leq& \frac{1}{K_{j}^\ast}\sum_{\ell=1}^{K_j^\ast}\sum_{h=-T_j}^{T_j} \vert \Gamma_j \left( h\right)\vert \vert h \vert O\left(\frac{1}{T_j}\right) = O\left(\frac{1}{T_j}\right).
\end{eqnarray*}
For $b_{122}$, let's focus on the case that $h\geq0$. Since $\frac{L_{1T_j}K_j}{T_j}=o(1)$,
\begin{eqnarray*}
    b_{122} &=& \left\vert\frac{1}{2K_j^\ast}%
\sum_{\ell=1}^{K_j^\ast}\sum_{h=0}^{L_{1T_j}}\left[  \int_0^{1-\frac{h}{T_j}} \cos\left( 4\pi \ell x \right)dx\right] \Gamma_j \left(h\right) \right\vert(1 + o(1)) \\
&\leq& O(1) \left\vert\frac{1}{K_j^\ast}%
\sum_{\ell=1}^{K_j^\ast}\sum_{h=0}^{L_{1T_j}} \sin\left( 4\pi \ell - \frac{4\pi \ell h}{T_j}  \right) \Gamma_j \left(h\right) \right\vert \\
&\leq& O(1) \frac{1}{K_j^\ast}%
\sum_{\ell=1}^{K_j^\ast}\sum_{h=0}^{L_{1T_j}}\left| \frac{4\pi \ell h}{T_j} + \left( \frac{4\pi \ell h}{T_j} \right)^2 \right| \left\vert  \Gamma_j \left(h\right) \right\vert \\
&\leq& O(1) \frac{1}{K_j^\ast}%
\sum_{\ell=1}^{K_j^\ast}\frac{\ell}{T_j}\sum_{h=0}^{L_{1T_j}} \vert \Gamma_j \left(h\right)\vert|h| \\
&=& O\left(\frac{K_j^\ast}{T_j} \right).
\end{eqnarray*}
We have the same result for the case that $h<0$.

Thus, we have
\begin{eqnarray} \label{eqn: B11}
    B_{11} = O\left(\frac{K_j^\ast}{T_j} \right).
\end{eqnarray}
 
For $B_{12},$ since $%
E\left( B_{12}\right) =0,$ it is sufficient to show that $Var\left(
B_{12}\right) =o\left( 1\right).$ Due to Assumption \ref{Assumption3: cumulant}, we have%
\begin{eqnarray*}
Var\left( B_{12}\right)  &=&\frac{1}{K_j^{\ast 2}}\sum_{\ell=1}^{K_j^\ast}\sum_{k =1}^{K_j^\ast}\frac{1}{T_j^{2}}\sum_{t=1}^{T_j}\sum_{s=1}^{T_j}%
\sum_{\tau =1}^{T_j}\sum_{w=1}^{T_j}\psi _{1,\ell}\left( \frac{t}{T_j}\right) \psi
_{1,\ell}\left( \frac{s}{T_j}\right) \psi _{1,k}\left( \frac{\tau }{T_j}\right)
\psi _{1,k}\left( \frac{w}{T_j}\right)  \\
&&\times \left[ E\left( u_{jt}u_{js}u_{j\tau }u_{jw}\right) -E\left(
u_{jt}u_{js}\right) E\left( u_{j\tau }u_{jw}\right) \right]  \\
&=&\frac{1}{K_j^{\ast 2}}\sum_{\ell=1}^{K_j^\ast}\sum_{k =1}^{K_j^\ast}\frac{%
1}{T_j^{2}}\sum_{t=1}^{T_j}\sum_{s=1}^{T_j}\sum_{\tau =1}^{T_j}\sum_{w=1}^{T_j}\psi
_{1,\ell}\left( \frac{t}{T_j}\right) \psi _{1,\ell}\left( \frac{s}{T_j}\right) \psi _{1,k}\left( \frac{\tau }{T_j}\right) \psi _{1,k}\left( \frac{w}{T_j}\right)  \\
&&\times \left[ E\left( u_{jt}u_{j\tau }\right) E\left( u_{js}u_{jw}\right)
+E\left( u_{jt}u_{jw}\right) E\left( u_{js}u_{j\tau }\right) \right] +O\left( 
\frac{1}{T_j}\right)  \\
&=&c_{1}+c_{2}+O\left( \frac{1}{T_j}\right) .
\end{eqnarray*}

Since $c_1$ and $c_2$ are identical, it is sufficient to show $c_1 = o(1)$ to prove $Var\left( B_{12}\right) =o(1)$. For $c_{1},$%
\begin{eqnarray*}
c_{1} &=&\frac{1}{K_j^{\ast 2}}\sum_{\ell=1}^{K_j^\ast}\sum_{k =1}^{K_j^\ast}\left( \frac{1}{T_j}\sum_{t=1}^{T_j}\sum_{\tau =1}^{T_j}\psi _{1,\ell}\left( \frac{t}{%
T_j}\right) \psi _{1,k}\left( \frac{\tau }{T_j}\right) E\left( u_{jt}u_{j\tau
}\right) \right) ^{2}  \notag \\
&=&\underset{=O\left( \frac{1}{K_j^\ast}\right) =o\left( 1\right) \text{ as 
}K_j^\ast\rightarrow \infty }{\underbrace{\frac{1}{K_j^{\ast 2}}%
\sum_{\ell=1}^{K_j^\ast}\left( \frac{1}{T_j}\sum_{t=1}^{T_j}\sum_{\tau =1}^{T_j}\psi
_{1,\ell}\left( \frac{t}{T_j}\right) \psi _{1,\ell}\left( \frac{\tau }{T_j}\right) E\left(
u_{jt}u_{j\tau }\right) \right) ^{2}}}  \notag \\
&&+\frac{1}{K_j^{\ast 2}}\sum_{\ell=1}^{K_j^\ast}\sum_{k\neq \ell}\left( \frac{1%
}{T_j}\sum_{t=1}^{T_j}\sum_{\tau =1}^{T_j}\psi _{1,\ell}\left( \frac{t}{T_j}\right) \psi
_{1,k }\left( \frac{\tau }{T_j}\right) E\left( u_{jt}u_{j\tau }\right) \right)
^{2} \\
&=& c_{11} + c_{12}.
\end{eqnarray*}%
It is straightforward that $c_{11} = O\left( \frac{1}{K_j^\ast} \right)$. Choose $L_{2T_j}$ that satisfies
\[
\frac{L_{2T_j}}{T_j} + \frac{1}{L_{2T_j}} = o(1)
\]
as $T_j$ grows. $c_{12}$ can be written as 
\begin{eqnarray*}
c_{12}&=&\frac{1}{K_j^{\ast 2}}\sum_{\ell=1}^{K_j^\ast}\sum_{k \neq \ell}\left(
\sum_{h=-L_{2T_j}}^{L_{2T_j}}\Gamma_j \left( h\right) \frac{1}{T_j}\sum_{1\leq
t,t+h\leq T}\psi _{1,\ell}\left( \frac{t}{T_j}\right) \psi _{1,k}\left( \frac{t+h%
}{T_j}\right) \right) ^{2} \\
&&+\frac{1}{K_j^{\ast 2}}\sum_{\ell=1}^{K_j^\ast}\sum_{k \neq \ell}\left(
\sum_{L_{2T_j}<\left\vert h\right\vert <T_j}\Gamma_j \left( h\right) \frac{1}{T_j}%
\sum_{1\leq t,t+h\leq T_j}\psi _{1,\ell }\left( \frac{t}{T}\right) \psi
_{1,k}\left( \frac{t+h}{T_j}\right) \right) ^{2} \\
&=&c_{121}+c_{122}.
\end{eqnarray*}%
For $c_{121},$ 
\begin{eqnarray*}
\left\vert c_{121}\right\vert  &\leq &\frac{1}{K_j^{\ast 2}}\sum_{\ell=1}^{K_j^\ast}\sum_{k \neq \ell}\left( \sum_{h=-L_{2T_j}}^{L_{2T_j}}\left\vert \Gamma_j \left(
h\right) \right\vert \left\vert \frac{1}{T_j}\sum_{1\leq t,t+h\leq T_j}\psi
_{1,\ell }\left( \frac{t}{T_j}\right) \left[ \psi _{1,k}\left( \frac{t}{T_j}\right)
+\psi_{1,k}^{\prime }\left( \frac{t_{1}}{T_j}\right) \frac{\left\vert
h\right\vert }{T_j}\right] \right\vert \right) ^{2}  \notag \\
&\leq &\frac{2}{K_j^{\ast 2}}\sum_{\ell=1}^{K_j^\ast}\sum_{k \neq \ell}\left(
\sum_{|h|=0}^{L_{2T_j}}\left\vert \Gamma_j \left( h\right) \right\vert
\left\vert \frac{1}{T_j}\sum_{1\leq t,t+|h|\leq T_j}\psi_{1,\ell}\left( \frac{t}{T_j}%
\right) \psi _{1,k}\left( \frac{t}{T_j}\right) \right\vert \right) ^{2}  \notag
\\
&&+\frac{2}{K_j^{\ast 2}}\sum_{\ell=1}^{K_j^\ast}\sum_{k \neq \ell}\left(
\sum_{|h|=0}^{L_{2T_j}} \left\vert
h\right\vert\left\vert \Gamma_j \left(
h\right) \right\vert \left\vert \frac{1}{T_j^2}\sum_{1\leq t,t+|h|\leq T}\psi
_{1,\ell }\left( \frac{t}{T_j}\right)\psi_{1,k}^{\prime }\left( \frac{t_{1}}{T_j}\right) \right\vert\right) ^{2}
\notag \\
&=& c_{1211} + c_{1212},
\end{eqnarray*}%
where $t_{1}\in \left( t,t+h\right)$. For $c_{1211}$, note that 
\begin{eqnarray*}
&& \sum_{|h|=0}^{L_{2T_j}}\left\vert \Gamma_j \left( h\right) \right\vert
\left\vert \frac{1}{T_j}\sum_{1\leq t,t+|h|\leq T_j}\psi _{1,\ell }\left( \frac{t}{T_j}%
\right) \psi _{1,k}\left( \frac{t}{T}\right) \right\vert  \\
&\leq &\underset{%
=O\left( 1\right) }{\underbrace{\sum_{|h|=0}^{L_{2T_j}}\left\vert \Gamma_j
\left( h\right) \right\vert }}\underset{=o\left( 1\right) \text{ by Lemma %
\ref{Lemma: Basis function}(ii)}}{\underbrace{\left\vert \int_{0}^{1-\frac{|h|}{T_j}}\psi _{1,\ell }\left( x\right) \psi
_{1,k}\left( x\right)dx \right\vert }}(1+o(1)) \\
&=&o\left( 1\right) 
\end{eqnarray*}%
with $L_{2T_j}=o\left( T_j\right)$ for all $k\neq \ell,$ which implies 
\[
c_{1211} = o(1).
\]

For $c_{1212}$, since $\psi _{1,k}^{\prime }\left( x\right) =2\pi k \sin
\left( 2\pi k x\right),$ 
\begin{eqnarray*}
c_{1212}&=&\frac{1}{K_j^{\ast 2}}\sum_{\ell=1}^{K_j^\ast}\sum_{k \neq \ell}\left(\sum_{h=-L_{2T_j}}^{L_{2T_j}} \left\vert h \right\vert\left\vert \Gamma_j \left(
h\right) \right\vert \left\vert \frac{1}{T_j^2}\sum_{1\leq t,t+h\leq T}\psi
_{1,\ell }\left(\frac{t}{T_j}\right)\psi_{1,k}^{\prime }\left( \frac{t_{1}}{T_j}\right) \right\vert \right)^2\\
&=&\frac{1}{K_j^{\ast 2}}\sum_{\ell=1}^{K_j^\ast}\sum_{k \neq \ell} \frac{k^2}{T_j^2}\left(\sum_{h=-L_{2T_j}}^{L_{2T_j}}\left%
\vert h\right\vert \left\vert \Gamma_j \left( h\right) \right\vert \left\vert \frac{2\pi}{T_j}\sum_{1\leq
t,t+h\leq T}\cos \left( 2\pi \ell \frac{t}{T_j}\right)  \sin
\left( \frac{2\pi k t_{1}}{T}\right) \right\vert \right)^2\\
&=& O\left( \frac{K_j^{\ast 2}}{T_j^2}\right). 
\end{eqnarray*}%

For $c_{122},$ we have%
\begin{eqnarray*}
c_{122} &\leq& \frac{1}{K_j^{\ast 2}}\sum_{\ell=1}^{K_j^\ast}\sum_{k \neq \ell}\left(
\sum_{L_{2T_j}<\left\vert h\right\vert <T_j}\Gamma_j \left( h\right) \frac{h}{L_{2T_j}} \frac{1}{T_j}%
\sum_{1\leq t,t+h\leq T_j}\psi _{1,\ell }\left( \frac{t}{T}\right) \psi
_{1,k}\left( \frac{t+h}{T_j}\right) \right) ^{2} \\
&\leq& \frac{1}{L_{2T_j}^2} \frac{1}{K_j^{\ast 2}}\sum_{\ell=1}^{K_j^\ast}\sum_{k \neq \ell}\left(
\sum_{h=-\infty}^{\infty} \left|\Gamma_j \left( h\right) \right| \left| h \right| \left(\frac{1}{T_j}%
\sum_{1\leq t,t+h\leq T_j}\psi _{1,\ell }\left( \frac{t}{T}\right) \psi
_{1,k}\left( \frac{t+h}{T_j}\right)\right) \right) ^{2} \\
&=& O\left(\frac{1}{L_{2T_j}^2}\right) = o(1).
\end{eqnarray*}%
as $L_{2T_j}\rightarrow \infty $ as $T_j\rightarrow \infty.$ 

Combining the results above, we have $c_1 = o(1)$ as $T_j , K_j^\ast \rightarrow \infty$ such that $K_j^\ast/T_j \rightarrow 0.$ Hence,
\[
B_{12} = o(1),
\]
which together with $B_{11}=O\left(\frac{K_j^\ast}{T_j}\right)$ in (\ref{eqn: B11}) implies 
\[
B_1 = o(1).
\]
It remains to show that $B_2 = o(1)$. Recall that
\[
B_2 = \frac{1}{T_j} \sum_{t=1}^{T_j} \sum_{s=1}^{T_j}  \frac{1}{K_j^\ast} \sum_{\ell=1}^{K_{j}^\ast} \psi_{2,\ell}\left( \frac{t}{T_j} \right)\psi_{2,\ell}\left( \frac{s}{T_j} \right)  u_{jt} u_{js} - \frac{1}{2}\Omega_{T_j}
\]
with $\psi_{2,\ell}(x) = \sin(2\pi \ell x)$. Note that $B_2$ is obtained from $B_1$ by replacing $\psi_{1,\ell}(x) = \cos(2\pi \ell x)$ with $\psi_{2,\ell}(x) = \sin(2\pi \ell x)$. The argument establishing $B_1 = o(1)$ carries over verbatim, so we omit the details.

Hence, we have established (i) and (ii), and therefore
\[
\hat{\Omega}_{\text{boot},T_j} \rightarrow^P \Omega_j.
\]
\end{proof}

\subsection{Additional simulation and empirical results}\label{app:extra-results}
This appendix collects results referenced in the main text: the scalar and multivariate size-adjusted power comparisons---which rely on infeasible oracle critical values and are therefore reported here rather than in the main text---the multivariate raw-power comparison, together with the component-wise working-from-home tests and the Ljung--Box diagnostics for the macroeconomic series.

\begin{table}[!t]
\centering
\caption{Multivariate raw power ($p=3$)}
\label{tab:mc_multivar_power}
\begin{tabular}{lllrrrrrrr}
\toprule
$T$ & $\rho$ & $c$ & SB size & HAR-$\chi^2$ & HAR-$F$ & SHAR-WB & IM(6) & IM(8) & IM(10) \\
\midrule
100 & 0.50 & 4.0 & 4.4 & 59.0 & 13.8 & 15.9 & 2.4 & 8.2 & 13.9 \\
100 & 0.50 & 8.0 & 4.4 & 93.4 & 43.0 & 51.1 & 15.0 & 45.0 & 63.2 \\
100 & 0.80 & 4.0 & 5.2 & 61.1 & 8.2 & 6.7 & 1.3 & 7.2 & 14.2 \\
100 & 0.80 & 8.0 & 5.2 & 79.6 & 16.0 & 11.9 & 7.7 & 25.1 & 39.4 \\
200 & 0.50 & 4.0 & 4.9 & 46.0 & 17.0 & 20.6 & 2.9 & 7.4 & 11.6 \\
200 & 0.50 & 8.0 & 4.9 & 91.2 & 62.1 & 65.4 & 15.4 & 40.9 & 57.5 \\
200 & 0.80 & 4.0 & 4.7 & 55.9 & 6.2 & 5.0 & 0.8 & 3.0 & 6.0 \\
200 & 0.80 & 8.0 & 4.7 & 78.0 & 11.9 & 9.5 & 5.1 & 15.2 & 24.8 \\
\bottomrule
\end{tabular}

\begin{minipage}{0.95\textwidth}
\footnotesize\emph{Notes:} Raw rejection frequencies (\%) under $\mu_1-\mu_2=cT^{-1/2}\iota_p$, unequal LRV, $p=3$. ``SB size'' is the SHAR-WB null rejection frequency for the corresponding $(T,\rho)$, reproduced from Table~\ref{tab:mc_multivar_size} to aid interpretation. IM($q$) is the grouped Hotelling test.
\end{minipage}
\end{table}

\begin{table}[!t]
\centering
\caption{Multivariate empirical size for the joint test ($p=3$, equal LRV)}
\label{tab:mc_multivar_size_equal}
\begin{tabular}{llrrrrrr}
\toprule
$T$ & $\rho$ & HAR-$\chi^2$ & HAR-$F$ & SHAR-WB & IM(6) & IM(8) & IM(10) \\
\midrule
  30 & 0.00 & 14.20 & 3.30 & 4.90 & 0.00 & 0.70 & 1.70 \\
  30 & 0.50 & 46.70 & 4.70 & 4.90 & 1.20 & 2.10 & 6.70 \\
  30 & 0.80 & 71.90 & 11.50 & 7.50 & 5.60 & 21.20 & 38.00 \\
  30 & 0.95 & 94.50 & 41.60 & 30.10 & 51.20 & 80.00 & 88.20 \\
  50 & 0.00 & 11.60 & 5.10 & 5.40 & 0.20 & 1.20 & 1.80 \\
  50 & 0.50 & 40.50 & 5.30 & 5.80 & 0.60 & 1.90 & 3.70 \\
  50 & 0.80 & 58.30 & 7.60 & 5.10 & 1.80 & 9.40 & 19.10 \\
  50 & 0.95 & 89.80 & 27.60 & 18.70 & 32.90 & 64.60 & 76.60 \\
  100 & 0.00 & 7.70 & 5.10 & 5.50 & 0.40 & 1.00 & 1.60 \\
  100 & 0.50 & 26.10 & 5.00 & 5.80 & 0.60 & 1.60 & 3.00 \\
  100 & 0.80 & 50.40 & 6.30 & 4.80 & 0.90 & 3.80 & 8.30 \\
  100 & 0.95 & 76.00 & 13.40 & 8.20 & 8.50 & 34.30 & 51.10 \\
  200 & 0.00 & 6.80 & 4.70 & 5.30 & 0.20 & 1.00 & 0.90 \\
  200 & 0.50 & 16.80 & 4.60 & 5.20 & 0.30 & 0.70 & 1.50 \\
  200 & 0.80 & 49.20 & 4.60 & 4.10 & 0.60 & 2.10 & 2.80 \\
  200 & 0.95 & 59.40 & 7.60 & 4.00 & 1.60 & 9.70 & 21.10 \\
  400 & 0.00 & 5.60 & 4.30 & 4.40 & 0.10 & 0.80 & 1.20 \\
  400 & 0.50 & 11.80 & 4.90 & 5.20 & 0.10 & 0.90 & 0.90 \\
  400 & 0.80 & 37.80 & 3.10 & 4.10 & 0.30 & 1.10 & 1.90 \\
  400 & 0.95 & 51.80 & 6.00 & 4.30 & 1.00 & 3.80 & 8.00 \\
\bottomrule
\end{tabular}

\begin{minipage}{0.95\textwidth}
\footnotesize\emph{Notes:} Rejection frequencies (\%) at the $5\%$ level for $H_0:\mu_1-\mu_2=0$, $p=3$, equal LRV ($D_1=D_2=I_p$); companion to the unequal-LRV panel in the main text. IM($q$) is the grouped two-sample Hotelling $T^2$ test, defined only for $q>p$.
\end{minipage}
\end{table}

\begin{table}[!t]
\centering
\caption{Scalar size-adjusted power}
\label{tab:mc_scalar_power}
\begin{tabular}{lllrrrrr}
\toprule
$T$ & $\rho$ & $c$ & HAR-$F$ & SHAR-WB & IM(6) & IM(8) & IM(10) \\
\midrule
100 & 0.50 & 4.0 & 29.3 & 27.8 & 30.3 & 31.7 & 32.4 \\
100 & 0.50 & 8.0 & 75.4 & 70.7 & 82.0 & 84.4 & 86.0 \\
100 & 0.80 & 4.0 & 14.2 & 13.9 & 15.7 & 16.6 & 16.6 \\
100 & 0.80 & 8.0 & 35.5 & 29.9 & 44.8 & 46.8 & 47.4 \\
200 & 0.50 & 4.0 & 30.9 & 27.0 & 30.9 & 32.4 & 33.5 \\
200 & 0.50 & 8.0 & 83.2 & 77.1 & 83.8 & 85.9 & 87.0 \\
200 & 0.80 & 4.0 & 12.6 & 10.6 & 14.5 & 13.5 & 13.8 \\
200 & 0.80 & 8.0 & 34.2 & 27.4 & 42.1 & 40.9 & 43.9 \\
\bottomrule
\end{tabular}

\begin{minipage}{0.95\textwidth}
\footnotesize\emph{Notes:} Size-adjusted rejection frequencies (\%) under the local alternative $\mu_1-\mu_2=cT^{-1/2}$, $p=1$. Size adjustment uses simulated null critical values, so every entry corresponds to a test with exactly $5\%$ size. IM($q$) is the scalar grouped $t$-test.
\end{minipage}
\end{table}

\begin{table}[!t]
\centering
\caption{Multivariate size-adjusted power ($p=3$)}
\label{tab:mc_multivar_power_adj}
\begin{tabular}{lllrrrrr}
\toprule
$T$ & $\rho$ & $c$ & HAR-$F$ & SHAR-WB & IM(6) & IM(8) & IM(10) \\
\midrule
100 & 0.50 & 4.0 & 18.0 & 18.4 & 23.0 & 25.6 & 29.8 \\
100 & 0.50 & 8.0 & 51.4 & 53.6 & 68.3 & 75.5 & 81.3 \\
100 & 0.80 & 4.0 & 6.3 & 6.5 & 11.2 & 10.8 & 9.0 \\
100 & 0.80 & 8.0 & 10.6 & 10.6 & 28.1 & 31.3 & 29.0 \\
200 & 0.50 & 4.0 & 19.4 & 21.1 & 19.9 & 21.5 & 22.3 \\
200 & 0.50 & 8.0 & 67.5 & 66.7 & 65.7 & 72.3 & 73.6 \\
200 & 0.80 & 4.0 & 5.7 & 5.2 & 7.2 & 7.2 & 8.3 \\
200 & 0.80 & 8.0 & 10.8 & 9.9 & 22.3 & 26.8 & 28.9 \\
\bottomrule
\end{tabular}

\begin{minipage}{0.95\textwidth}
\footnotesize\emph{Notes:} Size-adjusted rejection frequencies (\%) under $\mu_1-\mu_2=cT^{-1/2}\iota_p$, unequal LRV, $p=3$, using simulated null critical values.
\end{minipage}
\end{table}

\begin{table}[!t]
\centering
\caption{WFH application: scalar two-sample tests}
\label{tab:wfh_scalar_updated}
\begin{tabular}{lrrrrrr}
\toprule
Outcome & MeanDiff & Classical $p$ & Welch $p$ & HAR-$F$ $p$ & SHAR-WB $p$ & IM(8) $p$ \\
\midrule
Log calls & 0.056 & 0.062 & 0.021 & 0.091 & 0.065 & 0.228 \\
Log calls/min & 0.033 & 0.000 & 0.000 & 0.139 & 0.145 & 0.032 \\
Log phone time & 0.020 & 0.529 & 0.426 & 0.495 & 0.492 & 0.711 \\
Overall z & 0.071 & 0.228 & 0.184 & 0.425 & 0.492 & 0.554 \\
\bottomrule
\end{tabular}

\begin{minipage}{0.95\textwidth}
\footnotesize\emph{Notes:} Two-sided $p$-values. IM(8) is the grouped $t$-test of \citet{Ibragimov_Mueller2016} with eight consecutive blocks. SHAR-WB $p$-values are from the bootstrap distribution of the scalar statistic. Because the treatment series is short, the IM comparison should be read as diagnostic.
\end{minipage}
\end{table}

\begin{table}[!t]
\centering
\caption{FRED macro application: Ljung--Box diagnostics}
\label{tab:macro_lb_updated}
\begin{tabular}{lrrrr}
\toprule
Series & $Q(10)$ pre & $p$ pre & $Q(10)$ post & $p$ post \\
\midrule
Unemployment & 5783.28 & 0.000 & 1225.17 & 0.000 \\
Inflation & 5515.53 & 0.000 & 1180.69 & 0.000 \\
\bottomrule
\end{tabular}

\begin{minipage}{0.95\textwidth}
\footnotesize\emph{Notes:} Ljung--Box $Q(10)$ statistics and $p$-values for each component in the pre- and post-break subsamples.
\end{minipage}
\end{table}

\FloatBarrier

\bibliographystyle{apalike}

\clearpage
\bibliography{two_sample_bib}

\end{document}